\documentclass[12pt,reqno,a4paper]{amsart}
\usepackage[left=3cm, top=3.5cm, bottom=3.5cm, right=3cm]{geometry}
\usepackage{mathrsfs}
\usepackage{amsmath,amsxtra,amsfonts,amssymb, amsthm, amscd, epsfig}
\usepackage[dvipsnames,usenames]{color}
\usepackage{amssymb,amsmath,amsfonts,epsfig}
\usepackage{graphicx}

\numberwithin{equation}{section}

\newtheorem{thm}{Theorem}[section]
\newtheorem{Mthm}{Main Theorem}
\newtheorem{prop}[thm]{Proposition}
\newtheorem{cor}[thm]{Corollary}
\newenvironment{prfof}[1]{\noindent {\it Proof of #1} \ }{\hfill $\Box$}

\newtheorem{lem}[thm]{Lemma}

\newtheorem{rem}[thm]{Remark}

\newcommand{\eqa}{\begin{eqnarray}}
\newcommand{\eeqa}{\end{eqnarray}}
\newcommand{\beq}{\begin{equation}}
\newcommand{\eeq}{\end{equation}}
\newcommand{\nn}{\nonumber}
\newcommand{\p}{\partial}
\def \la {\langle}
\def \ra{\rangle}

 \def\res{\mathop{\text{\rm Res}}}

\def \dsum{\displaystyle\sum}

\allowdisplaybreaks[1]

\newcommand\Z{\mathbb{Z}}
\newcommand\C{\mathbb{C}}
\newcommand\mH{\mathcal{H}}
\newcommand\vp{\varphi}
\newcommand\ep{\epsilon}
\newcommand\om{\omega}
\newcommand\al{\alpha}
\newcommand\dt{\delta}
\newcommand\ld{\lambda}

\newcommand\sg{\sigma}
\def \la {\langle} \def \ra{\rangle}

\begin{document}

\title[]
{Infinite-dimensional Frobenius Manifolds Underlying the Universal Whitham Hierarchy }

\author[]{ Shilin Ma, Chao-Zhong Wu, Dafeng Zuo}

\address[]{Shilin Ma,  School of Mathematical Science,University of Science and Technology of China,
Hefei 230026, P.R.China  }
\email{mashilin@mail.ustc.edu.cn}

\address[]{Chao-Zhong Wu,  School of Mathematics,  Sun
Yat-Sen University,
Guangzhou 510275, P.R. China  }
 \email{wuchaozhong@sysu.edu.cn }

 \address[]{Dafeng Zuo,  School of Mathematical Science,University of Science and Technology of China,
Hefei 230026, P.R.China
 }
\email{dfzuo@ustc.edu.cn}

\date{\today}

\begin{abstract}
We construct a class of infinite-dimensional Frobenius
manifolds on the spaces of pairs of meromorphic functions with a pole at infinity and a movable pole. Such Frobenius
manifolds are shown to be underlying the universal Whitham hierarchy, which is an extension of the dispersionless
Kadomtsev-Petviashvili  hierarchy.

\vskip 2ex

\center{
\it \normalsize To the memory of Professor Boris Dubrovin
}
\end{abstract}

\maketitle 

\tableofcontents

\section{Introduction}

\subsection{Background} The concept of Frobenius manifold was introduced by Dubrovin
\cite{Du96} to exhibit the geometry behind the WDVV (Witten-Dijkgraaf-E. Verlinde-H. Verlinde)
system of nonlinear differential equations \cite{DVV, Witten} (see also \cite{Saito}). This concept plays a significant
role in several branches of mathematical physics, including the theory of singularities,
Gromov-Witten invariants and integrable systems e.t.c., see, for example, \cite{Du96, DZ2001, DSZZ2019, LRZ,Manin, MT2008, Zuo}
and references therein.

A Frobenius manifold of charge $d$ is an $n$-dimensional manifold $M$ equipped on each tangent space $T_v M$ with a
structure of Frobenius algebra $(A_v=T_v M, \circ, e, <~,~>)$ depending smoothly on $v\in M$, such that the
following three axioms are satisfied:
 \begin{itemize}
\item[(FM1)] The nondegenerate bilinear form $<~,~>$ is a flat metric on $M$, and the unity vector field $e$ is covariantly constant, i.e.,
$\nabla e=0$ (here $\nabla$ denotes the Levi-Civita connection for the flat metric);
\item[(FM2)] Let $c$ be a 3-tensor defined by $c(X,Y,Z):=<X\cdot Y, Z>$ with $X,\, Y,\,
Z\in T_v M$, then the 4-tensor $(\nabla_W c)(X,Y,Z)$ is symmetric in
$X,\, Y,\, Z, \, W \in T_v M$;
\item[(FM3)] There exists a vector field $E$, called the Euler vector field,
which satisfies the conditions $\nabla\nabla E=0$ and
$$
[E, X\circ Y] -[E,X]\circ Y -X\circ [E,Y] = X\circ Y,
$$
$$
E(<X,Y>)-<[E,X],Y>-<X,[E,Y]>=(2-d)<X,Y>
$$
for any vector fields $X, Y$ on $M$.
\end{itemize}

According to the axiom (FM1), on the Frobenius manifold $M$ one can choose a system of flat local coordinates $v=(v^1,\dots,v^n)$ such that the unity vector field reads $e=\dfrac{\p}{\p v^1}$.
By using such flat coordinates, a constant non-degenerate $n\times n$ matrix is given by
   \begin{align*}
     \eta_{\al\beta}=<\frac{\p}{\p v^\al},\frac{\p}{\p
     v^{\beta}}>,
   \end{align*}
and its inverse is denoted as $(\eta^{\al\beta})$.
The matrices $\eta_{\al\beta}$ and $\eta^{\al\beta}$ will be used to
lower and to lift indices, respectively, and summations over repeated Greek indices are assumed.
Let
\begin{equation*}
c_{\al\beta\gamma}=c\left(\frac{\p}{\p v^\al},\frac{\p}{\p v^\beta},
\frac{\p}{\p v^\gamma}\right),
\end{equation*}
then the product $\circ$ of the Frobenius algebra $T_v M$ is given by
   \begin{align*}
     \frac{\p}{\p v^{\al}}\circ \frac{\p}{\p v^{\beta}}=c^{\gamma}_{\al\beta}\frac{\p}{\p
     v^{\gamma}}, \quad
     c^{\gamma}_{\al\beta}:=\eta^{\gamma\ep}c_{\ep\al\beta}.
   \end{align*}
The structure constants of the product satisfy
\begin{equation}\label{WDVV1}
c_{1\al}^{\beta}=\delta_\al^\beta, \quad c_{\al\beta}^\ep
c_{\ep\gamma}^{\sg}=c_{\al\gamma}^\ep c_{\ep\beta}^{\sg}.
\end{equation}
According to the axioms (FM2) and (FM3), there locally exists a so-called potential function  $F(v)$ such that
\begin{align}
c_{\al\beta\gamma}=&\frac{\p^3 F}{\p v^{\al}\p v^{\beta}\p
v^{\gamma}}, \\
 \mathrm{Lie}_E F=&(3-d)F+\hbox{ quadratic terms in $v$}. \label{WDVV2}
\end{align}
In other words, the function $F$ solves the WDVV equation
\eqref{WDVV1}--\eqref{WDVV2}, and its third-order derivatives
$c_{\al\beta\gamma}$ are called the $3$-point correlation functions
in topological field theory \cite{Du92}.
Conversely, given a solution $F$ of the WDVV equation \eqref{WDVV1}--\eqref{WDVV2}
(including a flat metric, a unity vector field and a Euler vector field), one can
recover the structure of a Frobenius manifold.

Similar to the tangent space $T_v M$, the cotangent space $T_v^*M$ of $M$ also
carries a Frobenius algebra structure with a product $\star$ and  an invariant bilinear
form $<\, ,\,>^*$  given respectively by
\begin{align}
d v^\al\star d v^\beta=\eta^{\al\ep}c_{\ep\gamma}^{\beta} d v^\gamma, \quad <dt^\al,dt^\beta>^*=\eta^{\al\beta}. \nn
\end{align}
Moreover,  on $T_v^*M$ there is another symmetric
bilinear form, named as the \emph{intersection form}, defined by
\begin{equation*}
(d v^{\al},d v^{\beta})^*=g^{\al\beta}, \quad g^{\al\beta}:=i_{E}(d v^{\al}\circ d v^{\beta}).
\end{equation*}
It is known that the bilinear forms $g^{\al\beta}$ and $\eta^{\al\beta}$ on $T^*M$ compose a covariant flat pencil of metrics. Namely, any linear combination of these to metrics is still a flat metric, and their covariant Levi-Civita symbols obey the same linear combination relation.
Based on this fact and the pioneering work of Dubrovin and Novikov on Hamiltonian structures of hydrodynamics type \cite{DN},
Dubrovin established a link between Frobenius manifolds and certain $(1+1)$-dimensional dispersionless integrable
hierarchies. Such a seminal theory has been achieved in essentially the finite-dimensional case, namely, on an $n$-dimensional Frobenius manifold it is associated with a system of evolutionary equations of hydrodynamic type of $n$ unknown functions.

The first trial to extend
the theory of Frobenius manifolds to the case of infinite dimension was done by Carlet, Dubrovin and Mertens \cite{CDM}. More exactly, they discovered an infinite-dimensional Frobenius manifold structure on a space of pairs of certain meromorphic functions with single poles at the origin and at
infinity. Such a Frobenius manifold, in contrast to the finite-dimensional case, is associated with an integrable hierarchy of $2+1$ evolutionary equations \cite{CM, UT}.
Following this approach, more infinite-dimensional Frobenius manifolds were
constructed in \cite{WX, WZuo} by considering
pairs of meromorphic functions with higher-order poles at the origin and at
infinity, and such (formal) manifolds underly the
two-component BKP hierarchy \cite{DJKM-KPtype} and the Toda lattice hierarchy \cite{UT} respectively. In particular, when the meromorphic functions are with single poles at the origin and at infinity, the Frobenius manifold in \cite{WZ} is similar,
but not exactly the same, with that given in \cite{CDM}.

As it is known, a fundamental role in the theory of integrable systems is played by the Kadomtsev-Petviashvili (KP) hierarchy
\cite{DKJM-KPBKP} of $2+1$ evolutionary equations.
By using a method different from that in \cite{CDM}, Raimondo \cite{Ra} proposed an infinite-dimensional Frobenius manifold
for the dispersionless KP equation via the theory of Schwartz functions. What is more, a scheme was proposed by Szablikowski
\cite{Sz} to construct Frobenius manifolds based on the Rota-Baxter identity and a counterpart of the modified Yang-Baxter
 equation for the classical $r$-matrix, and this scheme was applied to a number of models including the dispersionless KP hierarchy.
  The relation between the constructions in \cite{Ra, Sz} and that in \cite{CDM, WX, WZuo} is not clear yet. {\it In fact, it is
   still open how to derive Frobenius manifolds underlying the dispersionless KP hierarchy via the approach originated in \cite{CDM}.}

Towards a possible solution to the above problem, the aim of the present paper is to study along the line of \cite{CDM, WX, WZuo} infinite-dimensional Frobenius
manifolds related to a certain extension of the dispersionless KP hierarchy. Such an extension (see \eqref{disphir1}--\eqref{disphir2} below for the definition) will be referred as the (special) universal Whitham hierarchy.
In fact, Krichever investigated in \cite{Kr88, Kr} the universal Whitham hierarchy and their reductions in moduli spaces of (formal) meromorphic functions on Riemann surfaces of all genera. In stead of the general setting in  \cite{Kr88, Kr},
this paper only concerns the case \eqref{disphir1}--\eqref{disphir2}, or precisely, its bi-Hamiltonian structures of hydrodynamic type derived in \cite{WZ}.

\subsection{Main results}
\label{sec-M}

Let $U$ be a neighborhood of $z=0$ on the Riemann sphere
$\mathbb{C}\cup\{\infty\}$, and we assume that there is one  circle $\Gamma$  around $U$  and another circle $\Gamma'$ outside $\Gamma$, which satisfy
\begin{equation}\label{GmGm}
|z_1|>|z_2| \hbox{ and } |z_1-\vp|>|z_2-\vp| \quad \hbox{ for any } \quad z_1\in\Gamma', \, z_2\in\Gamma, \, \vp\in U.
\end{equation}
For any $\varphi\in U$, we consider two sets of holomorphic functions on closed disks as follows:
\begin{align*}\label{}
\mH^-_\vp=&\left\{f(z)=\dsum_{i\ge0}f_i\,(z-\vp)^{-i}\mid f \hbox{ holomorphic outside
     } \Gamma \right\},\\
     \mH^+_\vp=&\left\{f(z)=\dsum_{i\ge0}f_i\,(z-\vp)^{i}\mid f \hbox{ holomorphic inside
     } \Gamma'  \right\}.
\end{align*}
Here by the words ``outside/inside'' we mean to include the boundaries $\Gamma$ and $\Gamma'$ respectively, namely,
the holomorphic functions can be extended analytically beyond such boundaries.

Given two arbitrary positive integers $m$ and $n$, let
\[
\tilde{\mathcal{M}}_{m,n}=\bigcup_{\vp\in U}\Big( \left(z^m+(z-\vp)^{m-2}\mH^-_\vp\right)\times (z-\vp)^{-n}\mH^+_\vp \Big),
\]
whose elements are written as pairs of Laurent series of the form
\begin{equation}\label{veca}
\vec{a}=\left(z^{m}+\sum_{i\leq m-2}a_{i}(z-\varphi)^{i},\sum_{i\geq -n}\hat{a}_{i}(z-\varphi)^{i}\right).
\end{equation}
For any point $\vec{a}=(a(z),\hat{a}(z))\in\tilde{\mathcal{M}}_{m,n}$, we introduce
\begin{equation}\label{zetaell}
\zeta(z)=a(z)-\hat{a}(z), \quad \ell(z)=a(z)_{+}+\hat{a}(z)_{-},
\end{equation}
where the subscripts `$\pm$' mean as before
to take the nonnegative and the negative powers in $(z-\varphi)$ (one notes that $z^m=\sum_{i=0}^m \binom{m}{i}\varphi^{m-i}(z-\varphi)^i$).
Clearly, the function $\zeta(z)$ is holomorphic in a neighborhood of the closed bend bounded by $\Gamma$ and $\Gamma'$, while $\ell(z)$ can be
extended analytically to $\C\setminus\{\vp\}$. Conversely, given such two functions $\zeta(z)$ and $\ell(z)$,  one can recover
$(a(z),\hat{a}(z))$ of the form \eqref{veca} by
\begin{equation}\label{}
a(z)=\zeta(z)_-+\ell(z), \quad \hat{a}(z)=-\zeta(z)_+ +\ell(z).
\end{equation}
Hence each point in $\tilde{\mathcal{M}}_{m,n}$ can be represented equivalently by the pair of series $(\zeta(z), \ell(z))$.

Let $\mathcal{M}_{m,n}$ be a subset of $\tilde{\mathcal{M}}_{m,n}$ consisting of points $\vec{a}=(a(z),\hat{a}(z))$ such
that the following conditions are fulfilled:
\begin{enumerate}
    \item[(C1)] the coefficient $\hat{a}_{-n}\ne 0$;
    \item[(C2)] at any point of the  circle $\Gamma$, it holds that
        \begin{equation}\label{M2}
        \zeta'(z)\ne0, \quad \ell'(z)\ne0, \quad a'(z)\hat{a}(z)-a(z)\hat{a}'(z)\ne0;
\end{equation}
    \item[(C3)] the function $\zeta(z)$ has  winding number 1 around $z=0$, such that it maps the circle $\Gamma$ biholomorphicly
    to a simple smooth curve $\Sigma$ around the origin.
\end{enumerate}
Such a subset $\mathcal{M}_{m,n}$ can be viewed as an infinite-dimensional manifold, whose coordinates can be chosen
as $\{a_{i}\}_{i\le m-2}\cup\{\hat{a}_{j}\}_{j\ge-n}\cup\{\varphi\}$.

On $\mathcal{M}_{m,n}$, let us introduce the set of variables
\begin{equation}\label{flat0}
\mathbf{t}=\{t_i\}_{i\in\Z}, \quad  \mathbf{h}=\big\{h_j\big\}_{j=1}^{m-1}, \quad \hat{\mathbf{h}}=\big\{\hat{h}_k\big\}_{k=0}^{n},
\end{equation}
where
\begin{align}\label{flatt0}
&t_i=\frac{1}{2\pi\mathbf{i}}\oint_{\Gamma}\frac{1}{i \zeta(z)^{i}}\,d z,
\quad i\in\mathbb{Z}\setminus \{0\};\quad
 t_0=\frac{1}{2\pi\mathbf{i}} \oint_{\Gamma}\log\frac{ z}{\zeta(z)}\,d z;
 \\
& h_j=-\frac{1}{j} \res_{z=\infty} \ell(z)^{\frac{j}{m}}
\,d z,\quad j=1,\cdots,m-1; \label{flath0}
\\
& \hat{h}_0=\vp;\quad
 \hat{h}_k=\frac{1}{k} \res_{z=\vp} \ell(z)^{\frac{k}{n}} \,d z,\quad
 k=1,\cdots,n. \label{flathh0}
\end{align}
\begin{Mthm}\label{main}
For any two positive integers $m$ and $n$, the set $\mathcal{M}_{m,n}$ is an  infinite-dimensional Frobenius manifold with a
system of flat coordinates  \eqref{flat0} such that
    \begin{itemize}
 \item[(i)] the unity vector field is
\begin{equation*}\label{vece}
\vec{e}=\left\{
\begin{aligned}
&\frac{\p}{\p t_{0}}+\frac{\p}{\p \hat{h}_{0}},\quad &m&=1,\\
&\frac{1}{m}\frac{\p}{\p h_{m-1}},\quad &m&\ge2;\\
\end{aligned}
\right.
\end{equation*}
\item[(ii)] the potential $\mathcal{F}$, up to quadratic terms in the flat coordinates, is
     \begin{align}
     \mathcal{F}=&\frac{1}{(2\pi \mathrm{i})^{2}}\oint_{\Gamma}\oint_{\Gamma'}\left(\frac{1}{2}\zeta(z_{1}) \zeta(z_{2})
     +\zeta(z_{1}) \ell(z_{2}) -\ell(z_{1}) \zeta(z_{2})\right) \log\left(\frac{z_{1}-z_{2}}{z_{1}}\right)dz_{1}dz_{2} \nn\\
&\qquad +\left(\frac{1}{2}t_{-1}-n \hat{h}_n\right)V_1(\mathbf{t})-t_{-1}V_2(\mathbf{h})+G(\mathbf{h},\hat{\mathbf{h}}). \label{potential0}
  \end{align}
  where the functions $V_1(\mathbf{t})$, $V_2(\mathbf{h})$ and $G(\mathbf{h},\hat{\mathbf{h}})$ are given in Lemma~\ref{thm-V1V2} below;
\item[(iii)] the Euler vector field is
  \begin{align*}
    \vec{E}=\sum_{i\in\mathbb{Z}}\left(\frac{1}{m}-i\right)t_{i}\frac{\p}{\p t_{i}}+\sum_{j=1}^{m-1} \frac{j+1}{m}  h_{j}\frac{\p}{\p h_{j}}+
\sum_{k=0}^{n}\left(\frac{1}{m}+\frac{k}{n}\right)\hat{h}_{k}\frac{\p}{\p \hat{h}_{k}},
\end{align*}
and the charge is $d=1-\frac{2}{m}$.
    \end{itemize}
\end{Mthm}

\begin{Mthm}\label{main2} The bi-Hamiltonian structure induced from the flat pencil of metrics on the Frobenius manifold $\mathcal{M}_{m,n}$
 coincides with that given in Proposition~\ref{thm-eKPham} below for the universal Whitham hierarchy \eqref{disphir1}--\eqref{disphir2}.
\end{Mthm}

\subsection{Organization} This paper is arranged as follows. In the next section, we are to prove Main Theorem 1, namely, to construct an infinite-dimensional Frobenius manifold structure on $\mathcal{M}_{m,n}$. In Section~3, we will study the relationship
between the infinite-dimensional Frobeniu manifold $\mathcal{M}_{m,n}$ and  the universal Whitham hierarchy (Main Theorem 2).
The last section is devoted to some concluding remarks.

\section{The infinite-dimensional Frobenius manifold structure on $\mathcal{M}_{m,n}$}
In this section, our aim is to construct an  infinite-dimensional
 Frobenius manifold structure on $\mathcal{M}_{m,n}$, including the flat metric, the potential function, the unity and the Euler vector fields.

\subsection{The tangent space and the cotangent space }

Let us describe the tangent and the cotangent spaces on
$\mathcal{M}_{m,n}$ with Laurent series. Firstly, at each point $\vec{a}=(a(z),\hat{a}(z))\in\mathcal{M}_{m,n}$ we identify a vector
$\p$ in the tangent space with its action $(\p
a(z), \p\hat{a}(z))$ and thus the tangent space is represented as
\begin{equation}\label{TM}
T_{\vec{a}}\mathcal{M}_{m,n}=(z-\varphi)^{m-2}\mH_\vp^- \times (z-\vp)^{-n-1}\mH_\vp^+.
\end{equation}
Accordingly, the cotangent space at $\vec{a}=(a(z),\hat{a}(z))$ is represented as the dual space of $T_{\vec{a}}\mathcal{M}_{m,n}$, $i.e.$,
\begin{equation}\label{TstaM}
T^{\ast}_{\vec{a}}\mathcal{M}_{m,n}=(z-\vp)^{-m+1}\mH_\vp^+\times (z-\varphi)^{n}\mH_\vp^-,
\end{equation}
with respect to the pairing
\begin{equation}\label{pairing}
\langle \vec{\omega},\vec{\xi}\rangle:=\frac{1}{2\pi\mathrm{i}}\oint_{\Gamma} \left(\omega(z)\xi(z)+\hat{\omega}(z)\hat{\xi}(z)\right)d z
\end{equation}
for $(\omega(z),\hat{\omega}(z))\in T^{\ast}_{\vec{a}}\mathcal{M}_{m,n}$ and  $(\xi(z),\hat{\xi}(z))\in T_{\vec{a}}\mathcal{M}_{m,n}$.

Let us fix a basis of $T^{\ast}_{\vec{a}}\mathcal{M}_{m,n}$ as
$$
\left\{ e_{i}^{\ast}=((z-\varphi)^{i},0), ~ \hat{e}_{j}^{\ast}=(0,(z-\varphi)^{j}) \mid  i\ge -m+1,~j\le n\right\},
$$
and introduce two generating functions
\begin{align}
 da(p):=\sum_{i\ge -m+1}(p-\varphi)^{-i-1}e^{\ast}_{i}=\left(\frac{(p-\varphi)^{m-1}}{(p-z)(z-\varphi)^{m-1}},0\right),
 &\quad |p-\varphi|>|z-\varphi|, \label{rpgener01} \\
 d\hat{a}(q):=\sum_{j\le n}(q-\varphi)^{-j-1}\hat{e}^{\ast}_{j}=\left(0,\frac{(z-\varphi)^{n+1}}{(z-q)(q-\varphi)^{n+1} } \right),
 & \quad |q-\varphi|<|z-\varphi|. \label{rpgener02}
\end{align}
The following lemma can be checked by using the Cauchy integral formula.
\begin{lem}\label{thm-dadah}
For any cotangent vector $\vec{\omega}=(\omega(z),\hat{\omega}(z))\in T^{\ast}_{\vec{a}}\mathcal{M}_{m,n}$ and tangent vector
$\vec{\xi}=(\xi(z),\hat{\xi}(z))\in T_{\vec{a}}\mathcal{M}_{m,n}$, it holds that
    $$
    \begin{aligned}
    &(\omega(z),0)=\frac{1}{2\pi\mathrm{i}}\oint_{|p-\varphi|>|z-\varphi|}\omega(p)da(p)\,d p,
\\
&(0,\hat{\omega}(z))=\frac{1}{2\pi\mathrm{i}} \oint_{|q-\varphi|<|z-\varphi|}\hat{\omega}(q)d\hat{a}(q)\,d q
    \end{aligned}
    $$
and
\begin{equation}\label{daxi}
    \langle da(p),\vec{\xi} \rangle=\xi(p),\quad \langle d\hat{a}(q),\vec{\xi} \rangle=\hat{\xi}(q).
\end{equation}
\end{lem}

On the cotangent space $T^{\ast}_{\vec{a}}\mathcal{M}_{m,n}$, let us introduce a symmetric bilinear form
\begin{equation}\label{blf}
\langle d\alpha(p),d\beta(q)\rangle^{\ast}:=\frac{\alpha'(p)}{p-q}+\frac{\beta'(q)}{q-p},\quad \alpha,\beta\in\{a,\hat{a}\},
\end{equation}
where the prime means to take the derivative with respect to the argument.
For the right hand side of \eqref{blf}, we remark that whenever $\al=\beta$,
the denominator $(p-q)$ is eliminated by a factor of
$\al'(p)-\al'(q)$; otherwise, suppose $\al=a$ and $\beta=\hat{a}$,
then $\frac{1}{p-q}$ is expanded according to $|p-\vp|>|q-\vp|$
as in \eqref{rpgener01}--\eqref{rpgener02}. Such kind of
convention will be used below.

Clearly, the pairing \eqref{pairing} is nondegenerate, hence there is a linear map
\begin{equation}\label{eta}
\eta: T^{\ast}_{\vec{a}}\mathcal{M}_{m,n}\to T_{\vec{a}}\mathcal{M}_{m,n}
\end{equation}
defined by
\begin{equation}\label{etadef}
\langle\vec{\omega}_{1},\eta\cdot\vec{\omega}_{2}\rangle=\langle \vec{\omega}_{1},\vec{\omega}_{2}\rangle^{\ast},\quad\forall \, \vec{\omega}_{1},\vec{\omega}_{2}\in T^{\ast}_{\vec{a}}\mathcal{M}_{m,n}.
\end{equation}
We proceed to describe the map $\eta$ more explicitly.
\begin{lem}\label{thm-etaom}
The linear map $\eta$ defined above can be represented as, for any cotangent vector $\vec{\omega}=(\omega(z), \hat{\omega}(z))\in T^{\ast}_{\vec{a}}\mathcal{M}_{m,n}$,
\begin{align}\label{}
    \eta\cdot\vec{\omega}=&\big(a'(z)[\omega(z)+\hat{\omega}(z)]_{-}-[\omega(z) a'(z)+\hat{\omega}(z)\hat{a}'(z)]_{-},\nn\\
    &-\hat{a}'(z)[\omega(z)+\hat{\omega}(z)]_{+}+[\omega(z) a'(z)+\hat{\omega}(z)\hat{a}'(z)]_{+}\big). \label{etaom} 
\end{align}
\end{lem}
\begin{proof}
Firstly, let us state the following two equalities that will be used below. For any function $f(z)=\sum_{i\in\Z}f_i (z-\vp)^i$
holomorphic on a neighbourhood of $\Gamma$, it holds that
    \begin{align}
&\frac{1}{2\pi\mathrm{i}}\oint_{|p-\varphi|>|q-\varphi|}\frac{1}{p-q} f(q)dq= f(p)_{-},\label{pqminus}\\
    &\frac{1}{2\pi\mathrm{i}}\oint_{|p-\varphi|>|q-\varphi|}
    \frac{1}{p-q} f(p)dp= f(q)_{+}. \label{pqplus}
    \end{align}
For $\vec{\omega}=(\omega(z),\hat{\omega}(z))\in T^{\ast}_{\vec{a}}\mathcal{M}_{m,n}$, we denote $\eta\cdot\vec{\omega}=(\xi(z),
 \hat{\xi}(z))\in T_{\vec{a}}\mathcal{M}_{m,n}$.
\begin{itemize}
  \item If $\vec{\omega}=(\omega(z),0)$, by using Lemma~\ref{thm-dadah} we have
    $$\begin{aligned}
    \xi(p)&=\langle da(p),(\omega(z),0)\rangle^{\ast}\\
    &=\frac{1}{2\pi\mathrm{i}}\oint_{|q-\varphi|>|z-\varphi|}\langle da(p),da(q)\rangle^{\ast}\omega(q)dq\\
&=\frac{1}{2\pi\mathrm{i}}\oint_{|p-\varphi|>|q-\varphi|} \left(\frac{a'(p)}{p-q}+\frac{a'(q)}{q-p}\right)\omega(q)dq\\
    &=a'(p)\omega(p)_{-}-(a'(p)\omega(p))_{-},\\
    \end{aligned}
    $$
    $$\begin{aligned}
    \hat{\xi}(q)&=\langle d\hat{a}(q),(\omega(z),0)\rangle^{\ast}\\
    &=\frac{1}{2\pi\mathrm{i}}\oint_{|p-\varphi|>|z-\varphi|}\langle d\hat{a}(q),da(p)\rangle^{\ast}\omega(p)dp\\
    &=\frac{1}{2\pi\mathrm{i}}\oint_{|p-\varphi|>|q-\varphi|} \left(\frac{\hat{a}'(q)}{q-p}+\frac{a'(p)}{p-q}\right)\omega(p)dp\\
    &=-\hat{a}'(q)\omega(q)_{+}+(a'(q)\omega(q))_{+}.
    \end{aligned}
    $$
Hence, we obtain
\[
\eta\cdot(\omega(z),0)=(a'(z)\omega(z)_{-}-(a'(z)\omega(z))_{-}, -\hat{a}'(z)\omega(z)_{+}+(a'(z)\omega(z))_{+}).
\]
  \item If $\vec{\omega}=(0,\hat{\omega}(z))$, similarly we have
\[
\eta\cdot(0,\hat{\omega}(z))=(a'(z)\hat{\omega}(z)_{-} -(\hat{a}'(z)\hat{\omega}(z))_{-},-\hat{a}'(z)\hat{\omega}(z)_{+} +(\hat{a}'(z)\hat{\omega}(z))_{+}).
\]
\end{itemize}
Thus the lemma is proved due to the linearity of $\eta$.
\end{proof}

\begin{lem} \label{thm-etabi}
The map $\eta$ in \eqref{eta} is a bijection.
\end{lem}
\begin{proof}
 It follows from Lemma~\ref{thm-etaom} that $\eta$ is surjective, so it suffices to show $\eta$ to be injective.
 For any $(\omega(z),\hat{\omega}(z))\in T^{\ast}_{\vec{a}}\mathcal{M}_{m,n}$, denote
 $(\xi(z),\hat{\xi}(z))=\eta\cdot(\omega(z),\hat{\omega}(z))$. With the use of \eqref{etaom}, one gets
\begin{align}
\xi(z)&=a'(z)[\omega(z)+\hat{\omega}(z)]_{-}-[\omega(z) a'(z)+\hat{\omega}(z)\hat{a}'(z)]_{-}, \label{xi}\\
\hat{\xi}(z)&=  -\hat{a}'(z)[\omega(z)+\hat{\omega}(z)]_{+}+[\omega (z)a'(z)+\hat{\omega}(z)\hat{a}'(z)]_{+} \label{xih}.
\end{align}
Furthermore, by using \eqref{zetaell} one has
$$
\xi(z)-\hat{\xi}(z)=(a'(z)-\hat{a}'(z))(\hat{\omega}(z)_{-}-\omega(z)_{+})=\zeta'(z)(\hat{\omega}(z)_{-}-\omega(z)_{+}).
$$
Since $\zeta'(z)\ne 0$ for $z\in\Gamma$ (recall the condition (C3) in Section~\ref{sec-M}), then
\begin{equation}\label{ompm1}
\omega(z)_{+}=-\left(\frac{\xi(z)-\hat{\xi}(z)}{\zeta'(z)}\right)_{+},\quad \hat{\omega}(z)_{-}=\left(\frac{\xi(z)-\hat{\xi}(z)}{\zeta'(z)}\right)_{-}.
\end{equation}

On the other hand, since
$$
\begin{aligned}
&\frac{1}{a'(z)}=\frac{1}{m}z^{-m+1}+\mathcal{O}((z-\varphi)^{-m})\in (z-\varphi)^{-m+1}\mH_\vp^-, \\
&\frac{1}{\hat{a}'(z)}=-\frac{1}{n\hat{a}_{-n}}(z-\varphi)^{n+1}+\mathcal{O}((z-\varphi)^{n+2})\in (z-\varphi)^{n+1}\mH_\vp^+,
\end{aligned}
$$
then from \eqref{xi} and \eqref{xih} it follows that
\begin{equation}\label{ompm2}
\left( [\omega(z)+\hat{\omega}(z)]_{-} \right)_{\ge-m+1}=\left(\frac{\xi(z)}{a'(z)}\right)_{\ge-m+1},
 \quad \left([\omega(z)+\hat{\omega}(z)]_{+}\right)_{\le n}=-\left(\frac{\hat{\xi}(z)}{\hat{a}'(z)}\right)_{\le n}.
\end{equation}
Here for a Laurent series $f(z)=\sum_{i\in\Z}f_i(z-\vp)^i$, the following notations are used:
\begin{equation}\label{}
f(z)_{\le k}=\sum_{i\le k}f_i(z-\vp)^i, \quad f(z)_{\ge k}=\sum_{i\ge k}f_i(z-\vp)^i.
\end{equation}
Recall $T^{\ast}_{\vec{a}}\mathcal{M}_{m,n}$ in \eqref{TstaM}, in combination of \eqref{ompm1} and \eqref{ompm2} one obtains
\begin{align}\label{}
\omega(z)_-=&\left( [\omega(z)+\hat{\omega}(z)]_{-} - \hat{\omega}(z)_{-} \right)_{\ge-m+1}=\left(\frac{\xi(z)}{a'(z)}
 - \left(\frac{\xi(z)-\hat{\xi}(z)}{\zeta'(z)}\right)_{-} \right)_{\ge-m+1}, \label{ompm3} \\
\hat\omega(z)_+=&\left( [\omega(z)+\hat{\omega}(z)]_{+} - \omega(z)_{+} \right)_{\le n}=\left(-\frac{\hat{\xi}(z)}{\hat{a}'(z)}
+ \left(\frac{\xi(z)-\hat{\xi}(z)}{\zeta'(z)}\right)_{+} \right)_{\le n}. \label{ompm4}
\end{align}
These equalities together with \eqref{ompm1} mean that the cotangent vector $(\omega(z),\hat\omega(z))$ can be uniquely solved.
Thus the proof of Lemma \ref{thm-etabi} is completed.
\end{proof}

From the above proof, one also knows that

\begin{cor}\label{thm-etainv}
For any vector $(\xi(z),\hat{\xi}(z))\in T_{\vec{a}}\mathcal{M}_{m,n}$, it holds that
\begin{equation}\label{etainv}
\eta^{-1}\left(\xi(z),\hat{\xi}(z)\right)=\left( \left(\frac{\xi(z)}{a'(z)} - \frac{\xi(z)-\hat{\xi}(z)}{\zeta'(z)}  \right)_{\ge-m+1},   \left(-\frac{\hat{\xi}(z)}{\hat{a}'(z)} +  \frac{\xi(z)-\hat{\xi}(z)}{\zeta'(z)}  \right)_{\le n}  \right).
\end{equation}
\end{cor}

\subsection{The flat metric}

With the help of the bijection $\eta$ in \eqref{eta}, let us introduce a bilinear form on the tangent space $T_{\vec{a}}\mathcal{M}_{m,n}$ as
\begin{equation}\label{metric}
\langle \p_1,\p_2\rangle_{\eta}:=\la
  \eta^{-1}(\p_1),\p_2\ra=\la \eta^{-1}(\p_1),\eta^{-1}(\p_2)\ra^\ast, \quad  \p_1,\p_2\in T_{\vec{a}}\mathcal{M}_{m,n}.
\end{equation}

\begin{lem}\label{thm-metric}For any vectors $\p_1, \p_2\in T_{\vec{a}}\mathcal{M}_{m,n}$,
the bilinear form \eqref{metric} is given by
\begin{equation}\label{}
 \langle \p_1,\p_2\rangle_{\eta}=
-\frac{1}{2\pi\mathrm{i}}\oint_{\Gamma}\frac{\p_{1}\zeta(z) \cdot\p_{2}\zeta(z)}{\zeta'(z)}dz
-\left(\res_{z=\infty}+\res_{z=\varphi}\right)\frac{\p_{1}\ell(z)\cdot \p_{2}\ell(z)}{\ell'(z)}dz.
\end{equation}
\end{lem}
\begin{proof}
Suppose that $\p_\nu$ are identified with $(\xi_\nu(z), \hat{\xi}_\nu(z))$ for $\nu=1,2$, it is easy to see
\begin{equation*}
\p_{\nu} \zeta(z)=\xi_{\nu}(z)-\hat{\xi}_{\nu}(z), \quad \p_{\nu} \ell(z)=\xi_{\nu}(z)_{+}+\hat{\xi}_{\nu}(z)_{-}, \qquad \nu=1,2.
\end{equation*}
Denote
\[
\eta^{-1}(\p_1)=(\om(z), \hat\om(z)),
\]
then, according to the definition \eqref{metric}, one has
$$
    \begin{aligned}
 \langle \p_1,\p_2\rangle_{\eta}&=
    \frac{1}{2\pi\mathrm{i}}\oint_{\Gamma}\left(\omega(z)\xi_{2}(z)+\hat{\omega}(z)\hat{\xi}_{2}(z)\right)dz=\mathcal{I}_{1}+\mathcal{I}_{2}+\mathcal{I}_{3},
    \end{aligned}
$$
    where
$$
    \begin{aligned}
\mathcal{I}_{1}&=\frac{1}{2\pi\mathrm{i}}\oint_{\Gamma}\left(\omega(z)_{+}-\hat{\omega}(z)_{-}\right)\left(\xi_{2}(z)-\hat{\xi}_{2}(z)\right)dz,\\
\mathcal{I}_{2}&=\frac{1}{2\pi\mathrm{i}}\oint_{\Gamma}\left(\omega(z)_{-}+\hat{\omega}(z)_{-}\right)\xi_{2}(z)dz,\\
\mathcal{I}_{3}&=\frac{1}{2\pi\mathrm{i}}\oint_{\Gamma}\left(\omega\left(z\right)_{+}+\hat{\omega}\left(z\right)_{+}\right)\hat{\xi}_{2}\left(z\right)dz.
    \end{aligned}
$$
Let us proceed to calculate these integrals. Firstly, by using \eqref{ompm1} one has
    $$\begin{aligned}
\mathcal{I}_{1}&=-\frac{1}{2\pi\mathrm{i}} \oint_{\Gamma}\frac{(\xi_{1}(z)-\hat{\xi}_{1}(z))(\xi_{2}(z)-\hat{\xi}_{2}(z))}{\zeta'(z)}d z
    =-\frac{1}{2\pi\mathrm{i}}\oint_{\Gamma}\frac{\p_{1}\zeta(z)\cdot \p_{2}\zeta(z)}{\zeta'(z)}d z.
    \end{aligned}
    $$
Secondly, observe that when $|z|\to\infty$ and $\nu=1,2$, one has
\begin{align*}
\p_\nu a(z), \p_\nu \ell(z)\in (z-\vp)^{m-2}\mH_\vp^-;& \quad \p_\nu a(z)-\p_\nu \ell(z)\in (z-\vp)^{-1}\mH_\vp^-; \\
\frac{1}{a'(z)}, \frac{1}{\ell'(z)} \in (z-\vp)^{-m+1}\mH_\vp^-;& \quad \frac{1}{a'(z)}-\frac{1}{\ell'(z)} \in (z-\vp)^{-2 m+1}\mH_\vp^-.
\end{align*}
With the use of \eqref{ompm2} and $\ell(z)=a(z)_++\hat{a}(z)_-$,  one obtains
\begin{align*}
\mathcal{I}_{2}&=\frac{1}{2\pi\mathrm{i}}\oint_{\Gamma} \left( [\omega(z)+\hat{\omega}(z)]_{-} \right)_{\ge-m+1}\xi_{2}(z)\,dz \\
&=\frac{1}{2\pi\mathrm{i}}\oint_{\Gamma} \left(\frac{\p_1 a(z)}{a'(z)}\right)_{\ge-m+1}\p_2 a(z)\,dz \\
&=\frac{1}{2\pi\mathrm{i}}\oint_{\Gamma} \left(\frac{\p_1 \ell(z)}{a'(z)}\right)_{\ge-m+1}\p_2 \ell(z) \,dz \\
&=\frac{1}{2\pi\mathrm{i}}\oint_{\Gamma} \left(\frac{\p_1 \ell(z)}{\ell'(z)}\right)_{\ge-m+1}\p_2 \ell(z) \,dz \\
&=-\res_{z=\infty}\frac{\p_{1}\ell(z)\cdot \p_{2}\ell(z)}{\ell'(z)}\,dz.
\end{align*}
Similarly, when $|z-\vp|\to0$ and $\nu=1,2$, one has
\begin{align*}
\p_\nu \hat{a}(z), \p_\nu \ell(z)\in (z-\vp)^{-n-1}\mH_\vp^+;& \quad \p_\nu \hat{a}(z)-\p_\nu \ell(z)\in \mH_\vp^+; \\
\frac{1}{\hat{a}'(z)}, \frac{1}{\ell'(z)} \in (z-\vp)^{n+1}\mH_\vp^+;& \quad \frac{1}{\hat{a}'(z)}-\frac{1}{\ell'(z)} \in (z-\vp)^{2n+2}\mH_\vp^+
\end{align*}
and
\begin{align*}
\mathcal{I}_{3}&=\frac{1}{2\pi\mathrm{i}}\oint_{\Gamma} \left([\omega(z)+\hat{\omega}(z)]_{+}\right)_{\le n} \hat{\xi}_{2}(z)\,dz\\
&=-\frac{1}{2\pi\mathrm{i}}\oint_{\Gamma} \left(\frac{\p_1\hat{a}(z)}{\hat{a}'(z)}\right)_{\le n} \p_2\hat{a}(z)\,dz\\
&=-\frac{1}{2\pi\mathrm{i}}\oint_{\Gamma} \left(\frac{\p_1\ell(z)}{\hat{a}'(z)}\right)_{\le n} \p_2\ell(z)\,dz\\
&=-\frac{1}{2\pi\mathrm{i}}\oint_{\Gamma} \left(\frac{\p_1\ell(z)}{\ell'(z)}\right)_{\le n} \p_2\ell(z)\,dz\\
&=-\res_{z=\varphi}\frac{\p_{1}\ell(z)\cdot \p_{2}\ell(z)}{\ell'(z)}\,dz.
    \end{align*}
Therefore, the lemma is proved.
\end{proof}

The bilinear form \eqref{metric} can be viewed as a metric on the manifold $\mathcal{M}_{m,n}$ (not necessary to be positively definite).
We proceed to show that this metric is a flat one.
According to the condition (C3) in Section~\ref{sec-M}, we can consider the inverse function of $\zeta(z)$, $i.e.$,
\[
z=z(\zeta): \Sigma\to\Gamma.
\]
This function can be extended holomorphicly to a neighborhood of $\Sigma$ that surrounds $\zeta=0$. Let us take the Laurent expansion
\begin{equation}\label{zzeta}
{z}(\zeta)=\sum_{i\in\Z}t_{i}\zeta^{i}, \quad \zeta\in\Sigma,
\end{equation}
where
\begin{equation}\label{ti}
t_{i}=\frac{1}{2\pi\mathrm{i}}\oint_{\Sigma} {z}(\zeta)\zeta^{-i-1}d\zeta, \quad
  i\in \mathbb{Z}.
\end{equation}
On the other hand, we introduce the following two functions in a punctured neighborhood of $\infty$ and of $\vp$ respectively,
\begin{align}\label{chi}
&\chi(z):=\ell(z)^{\frac{1}{m}}=z+\chi_{1}z^{-1}+\chi_{2}z^{-2}+\cdots \quad \hbox{near}\quad  \infty; 
\\
&\hat{\chi}(z):=\ell(z)^{\frac{1}{n}}={\hat{a}_{-n}}^{\frac{1}{n}}(z-\vp)^{-1}+\hat{\chi}_{0} +\hat{\chi}_{1}(z-\vp)
+\cdots \quad \hbox{near}\quad \vp. \label{chih}
\end{align}
The inverse functions of them can be represented as
\begin{align}\label{}
z(\chi)=&\chi-h_{1}\chi^{-1}-h_{2}\chi^{-2}-\cdots-h_{m-1}\chi^{-m+1}+\mathcal{O}(\chi^{-m}), \quad |\chi|\to\infty; \\
z(\hat{\chi})=& \hat{h}_0+\hat{h}_{1}\hat{\chi}^{-1}+\hat{h}_{2}\hat{\chi}^{-2} +\cdots+\hat{h}_{n}\hat{\chi}^{-n}
+\mathcal{O}(\hat{\chi}^{-n-1}), \quad |\hat\chi|\to\infty.
\end{align}
In particular, one has $\hat{h}_0=\vp$. It is easy to see that the  correspondence
\begin{equation}\label{}
(\zeta(z), \ell(z)) \longleftrightarrow \{t_{i}\}_{i\in \mathbb{Z}}\cup\{h_{j}\}_{j=1}^{m-1}\cup\{\hat{h}_{k}\}_{k=0}^{n}
\end{equation}
is one-to-one.
In other words, we obtain a system of coordinates of $\mathcal{M}_{m,n}$. Clearly, such coordinates can be represented equivalently as
 \eqref{flatt0}--\eqref{flathh0}, and in what follows we will use the notations $\mathbf{t}$, $\mathbf{h}$ and $\hat{\mathbf{h}}$ given in \eqref{flat0}.
Moreover, unless
otherwise stated, the following convention of indices will be assumed below
\[
i, i_1, i_2, i_3\in\Z,\quad j, j_1,j_2,j_3\in\{1,2,\ldots,m-1\},
\quad k,k_1,k_2,k_3\in\{0,1,\ldots,n\}.
\]

\begin{prop}\label{thm-flatmetric}
The metric defined by \eqref{metric} is flat, and $\mathbf{t}\cup\mathbf{h}\cup\hat{\mathbf{h}}$ of $\mathcal{M}_{m,n}$ serves as a system of flat coordinates. More precisely, these flat coordinates satisfy
\begin{align}
    \left\langle \frac{\p}{\p t_{i_{1}}},\frac{\p}{\p t_{i_{2}}}\right\rangle_{\eta}&= -\delta_{-1,i_{1}+i_{2}},\\
    \left\langle \frac{\p}{\p h_{j_{1}}},\frac{\p}{\p h_{j_{2}}}\right\rangle_{\eta}&= m\delta_{m,j_{1}+j_{2}},\\
    \left\langle \frac{\p}{\p \hat{h}_{k_{1}}},\frac{\p}{\p \hat{h}_{k_{2}}}\right\rangle_{\eta}&= n\delta_{n,k_{1}+k_{2}},
\end{align}
and
\begin{align}
    \left\langle \frac{\p}{\p t_{i}},\frac{\p}{\p h_{j}}\right\rangle_{\eta}&=\left\langle
    \frac{\p}{\p t_{i }},\frac{\p}{\p \hat{h}_{k}}\right\rangle_{\eta} =
    \left\langle \frac{\p}{\p h_{j }},\frac{\p}{\p \hat{h}_{k }}\right\rangle_{\eta}=0.
    \end{align}
\end{prop}
\begin{proof}
For preparation, let us show the following equalities:
\begin{align}
&\frac{\p\zeta(z)}{\p t_{i}}=-\zeta(z)^{i}\zeta'(z), \quad \frac{\p\ell(z)}{\p t_{i}}=0, \label{zetaellt}\\
&\frac{\p\zeta(z)}{\p h_{j}}=0, \quad  \frac{\p\ell(z)}{\p h_{j}}=(\ell'(z)\chi(z)^{-j})_{+}, \label{zetaellh}\\
&\frac{\p\zeta(z)}{\p \hat{h}_{k}}=0, \quad \frac{\p\ell(z)}{\p \hat{h}_{k}}=-(\ell'(z)\hat{\chi}(z)^{-k})_{-}. \label{zetaellhh}
\end{align}
Firstly, the inverse function  \eqref{zzeta} of $\zeta(z)$ satisfies $z=\left.z(\zeta)\right|_{\zeta\mapsto\zeta(z)}$.
 Taking its derivative with respective to $t_i$, one has
\begin{align*}
0=\left.\frac{\p z(\zeta)}{\p
t_i}\right|_{\zeta\mapsto\zeta(z)}+\left.z'(\zeta)\right|_{\zeta\mapsto\zeta(z)}
\frac{\p \zeta(z)}{\p t_i},
\end{align*}
which leads to the first equality in \eqref{zetaellt}, namely,
\[
\frac{\p \zeta(z)}{\p t_i}=-\left.\frac{1}{z'(\zeta)}\right|_{\zeta\mapsto\zeta(z)}\zeta(z)^i=-\zeta'(z)\zeta(z)^i.
\]
Similarly, from \eqref{chi} one has
\begin{align*}
& \frac{\p \chi(z)}{\p
h_j}=\chi'(z)\cdot\left.\left(\chi^{-j}+\mathcal{O}(\chi^{-m})\right)\right|_{\chi\mapsto\chi(z)}
=\chi'(z)\left(\chi(z)^{-j}+\mathcal{O}(z^{-m})\right), \quad |z|\to\infty,
\end{align*}
hence
\begin{align*}
&\frac{\p \ell(z)}{\p h_j}=\frac{\p \ell(z)_+}{\p h_j}=
\left(m \chi(z)^{m-1}\frac{\p\chi(z)}{\p h_j} \right)_+ = \left(m \chi(z)^{m-1-j} \chi'(z)\right)_+=\left( \ell'(z)  \chi(z)^{-j}\right)_+,
\end{align*}
which is the second equality in \eqref{zetaellh}.
In the same way when $|z-\vp|\to0$, for \eqref{chih} one gets
\begin{align*}
& \frac{\p \hat\chi(z)}{\p
\hat{h}_k}= -\hat{\chi}'(z)\cdot\left.\left(\hat{\chi}^{-k}+O(\hat{\chi}^{-n-1})\right)\right|_{\hat{\chi}\mapsto\hat{\chi}(z)} =-\hat{\chi}'(z)\left(\hat{\chi}(z)^{-k}+O((z-\vp)^{n+1})\right),
\end{align*}
hence
\begin{align*}
&\frac{\p \ell(z)}{\p \hat{h}_k}=\frac{\p \ell(z)_-}{\p \hat{h}_k}=
\left(n \hat{\chi}(z)^{n-1}\frac{\p\hat{\chi}(z)}{\p \hat{h}_k} \right)_- = -\left(n \hat{\chi}(z)^{n-1-k} \hat{\chi}'(z)\right)_-=-\left( \ell'(z)  \hat{\chi}(z)^{-k}\right)_-,
\end{align*}
which is just the second equality in \eqref{zetaellhh}. The other cases in \eqref{zetaellt}--\eqref{zetaellhh} are trivial.

Now we are ready to verify the proposition based on Lemma~\ref{thm-metric}. We have that
\begin{align*}
    \left\la \frac{\p}{\p t_{i_{1}}},\frac{\p}{\p t_{i_{2}}}\right\ra_{\eta}&=-\frac{1}{2\pi\mathrm{i}}\oint_{\Gamma}\frac{\zeta^{i_{1}+i_{2}}\zeta'\zeta'}{\zeta'}dz=-\delta_{-1,i_{1}+i_{2}},\\
    \left\la \frac{\p}{\p h_{j_{1}}},\frac{\p}{\p h_{j_{2}}}\right\ra_{\eta}&=-\res_{z=\infty}\frac{(\ell'\chi^{-j_{1}})_{+}(\ell'\chi^{-j_{2}})_{+}}{\ell'}dz\\
    &=-\res_{z=\infty}\frac{(\ell'\chi^{-j_{1}}-(\ell'\chi^{-j_{1}})_{-})(\ell'\chi^{-j_{2}}-(\ell'\chi^{-j_{2}})_{-})}{\ell'}dz\\
    &=-\res_{z=\infty}\chi^{-j_{1}}(\ell'\chi^{-j_{2}})dz\\
    &=-m\res_{z=\infty}\chi^{-j_{1}}\chi^{m-1}\chi'\chi^{-j_{2}}dz=m\delta_{m,j_{1}+j_{2}},\\
    \left\la \frac{\p}{\p \hat{h}_{k_{1}}},\frac{\p}{\p \hat{h}_{k_{2}}}\right\ra_{\eta}&=-\res_{z=\vp}\frac{(\ell'\hat{\chi}^{-k_{1}})_{-}(\ell'\hat{\chi}^{-k_{2}})_{-}}{\ell'}dz\\
    &=-\res_{z=\vp}\frac{(\ell'\hat{\chi}^{-k_{1}}-(\ell'\hat{\chi}^{-k_{1}})_{+})(\ell'\hat{\chi}^{-k_{2}}-(\ell'\hat{\chi}^{-k_{2}})_{+})}{\ell'}dz\\
    &=-\res_{z=\vp}\hat{\chi}^{-k_{1}}(\ell'\hat{\chi}^{-k_{2}})dz\\
    &=-n\res_{z=\vp}\hat{\chi}^{-k_{1}}\hat{\chi}^{n-1}\hat{\chi}'\hat{\chi}^{-k_{2}}dz=n\delta_{n,k_{1}+k_{2}},
    \end{align*}
and that all other pairings of such vectors vanish. Thus the proposition is proved.
\end{proof}

The above proof together with \eqref{zetaell} also implies
\begin{cor}\label{thm-flatv}
In the tangent space $T_{\vec{a}}\mathcal{M}_{m,n}$, there is a basis consisting of vectors  represented as
\begin{align}
\frac{\p}{\p t_{i}}=(-(\zeta(z)^{i}\zeta'(z))_{-},(\zeta(z)^{i}\zeta'(z))_{+}),&\quad i\in\mathbb{Z}; \label{flatv1} \\
\frac{\p}{\p h_{j}}=((\ell'(z)\chi(z)^{-j})_{+},(\ell'(z)\chi(z)^{-j})_{+}),&\quad 1\le j\le m-1;\\
\frac{\p}{\p\hat{h}_{k}}=(-(\ell'(z)\hat{\chi}(z)^{-k})_{-},-(\ell'(z)\hat{\chi}(z)^{-k})_{-}),&\quad 0\le k\le n. \label{flatv3}
\end{align}
\end{cor}

Furthermore, this result together with Corollary~\ref{thm-etainv} leads to
\begin{cor}\label{thm-flatcov}
    In the cotangent space $T_{\vec{a}}^\ast\mathcal{M}_{m,n}$, there is a basis consisting of vectors represented as
    \begin{align}
    \eta^{-1} \frac{\p}{\p t_{i}}=((\zeta(z)^{i})_{\ge -m+1},-(\zeta(z)^{i})_{\le n}),&\quad i\in\mathbb{Z}; \label{flatcov1}\\
    \eta^{-1}\frac{\p}{\p h_{j}}=( (\chi(z)^{-j})_{\ge -m+1},0),&\quad 1\le j\le m-1; \label{flatcov2} \\
    \eta^{-1}\frac{\p}{\p \hat{h}_{k}}=(0,(\hat{\chi}(z)^{-k})_{\le n}),&  \quad 0\le k\le n. \label{flatcov3}
    \end{align}
\end{cor}
\begin{proof}
We substitute \eqref{flatv1}--\eqref{flatv3} into the formula \eqref{etainv} and observe  the following equalities:
\begin{align*}
&\left(\frac{f(z)_-}{a'(z)}\right)_{\ge-m+1}=0, \quad \left(\frac{f(z)_+}{\hat{a}'(z)}\right)_{\le n}=0 \quad \hbox{with arbitrary}
\quad f(z)=\sum_{i\in\Z}f_i (z-\vp)^i; \\
&(\ell'(z)\chi(z)^{-j})_{+}=(a'(z)\chi(z)^{-j})_{+}, \quad (\ell'(z)\hat{\chi}(z)^{-k})_{-}=(\hat{a}'(z)\hat{\chi}(z)^{-k})_{-},
\end{align*}
then the corollary is achieved after a straightforward calculation.
\end{proof}

\subsection{The Frobenius algebra structure}

For any $\vec{a}=(a(z),\hat{a}(z))\in\mathcal{M}_{m,n}$, let us introduce a product on the cotangent space $T^{\ast}_{\vec{a}}\mathcal{M}_{m,n}$ as
\begin{equation}\label{prodcot}
d\alpha(p)\star d\beta(q)=\frac{\beta'(q)}{q-p} d\alpha(p)+\frac{\alpha'(p)}{p-q} d\beta(q),\quad \alpha, \beta\in \{a,\hat{a}\}.
\end{equation}
\begin{lem}\label{thm-prodcot}
The product on $T^{\ast}_{\vec{a}}\mathcal{M}_{m,n}$ defined by \eqref{prodcot} is commutative, associative and invariant with respect to
the bilinear form given in \eqref{blf}.
\end{lem}
\begin{proof}
The product in \eqref{prodcot} is clearly commutative. What is more, for $\alpha_1, \alpha_2, \alpha_3\in\{a, \hat{a}\}$,  one has
    $$ \begin{aligned}
    &d\alpha_{1}(p_{1})\star d\alpha_{2}(p_{2})\star d\alpha_{3}(p_{3}) \\=&\left(\frac{\alpha_{1}'(p_{1})}{p_{1}-p_{2}}d\alpha_{2}(p_{2})
     +\frac{\alpha_{2}'(p_{2})}{p_{2}-p_{1}}d\alpha_{1}(p_{1})\right)\star d\alpha_{3}(p_{3})\\
    =&\frac{\alpha_{1}'(p_{1})}{p_{1}-p_{2}}\frac{\alpha_{2}'(p_{2})}{p_{2}-p_{3}}d\alpha_{3}(p_{3})
    +\frac{\alpha_{1}'(p_{1})}{p_{1}-p_{2}}\frac{\alpha_{3}'(p_{3})}{p_{3}-p_{2}}d\alpha_{2}(p_{2})\\
    &+\frac{\alpha_{2}'(p_{2})}{p_{2}-p_{1}}\frac{\alpha_{1}'(p_{1})}{p_{1}-p_{3}}d\alpha_{3}(p_{3})
    +\frac{\alpha_{2}'(p_{2})}{p_{2}-p_{1}}\frac{\alpha_{3}'(p_{3})}{p_{3}-p_{1}}d\alpha_{1}(p_{1})\\
 =&\frac{\{\alpha_{1}'(p_{1})\alpha_{2}'(p_{2})d\alpha_{3}(p_{3})(p_{1}-p_{2})
 +\hbox{c.p.}(1,2,3)\}}{(p_{1}-p_{2})(p_{2}-p_{3})(p_{3}-p_{1})},
    \end{aligned}
$$
where $\hbox{c.p.}$ means ``cyclic permutation''. So this product is associative. Finally, by using \eqref{blf} it is straightforward to verify
\begin{align*}
&\left\langle  d\alpha_{1}(p_{1})\star d\alpha_{2}(p_{2}), d\alpha_{3}(p_{3}) \right\rangle^{\ast} \\=&
\left\langle  \frac{\alpha_{1}'(p_{1})}{p_{1}-p_{2}}d\alpha_{2}(p_{2}) +\frac{\alpha_{2}'(p_{2})}{p_{2}-p_{1}}d\alpha_{1}(p_{1}),
d\alpha_{3}(p_{3}) \right\rangle^{\ast} \\
=& \frac{\alpha_{1}'(p_{1})}{p_{1}-p_{2}}\left(\frac{\alpha_{2}'(p_{2})}{p_{2}-p_{3}} +\frac{\alpha_{3}'(p_{3})}{p_{3}-p_{2}} \right)
+\frac{\alpha_{2}'(p_{2})}{p_{2}-p_{1}} \left(\frac{\alpha_{1}'(p_{1})}{p_{1}-p_{3}} +\frac{\alpha_{3}'(p_{3})}{p_{3}-p_{1}} \right) \\
=& \frac{\alpha_{1}'(p_{1})\alpha_{2}'(p_{2}) }{(p_{1}-p_{3})(p_2-p_3)}+\hbox{c.p.}(1,2,3),
\end{align*}
which implies the invariance of the product. The lemma is proved.
\end{proof}

\begin{lem}\label{thm-omprod}
The product defined by \eqref{prodcot} can be represented as follows: for any $\vec{\omega}_{1}
=(\omega_{1}(z),\hat{\omega}_{1}(z))$, $\vec{\omega}_{2}=(\omega_{2}(z),\hat{\omega}_{2}(z))\in T^{\ast}_{\vec{a}}\mathcal{M}_{m,n}$,
\begin{align}
    \vec{\omega}_{1}\star \vec{\omega}_{2}=(&[\omega_{2}(\omega_{1}a')_{+}
    -\omega_{1}(\omega_{2}a')_{-}-\omega_{2}(\hat{\omega}_{1}\hat{a}')_{-}-\omega_{1}(\hat{\omega}_{2}\hat{a}')_{-}]_{\ge -m+1}, \nn\\
    &[\hat{\omega}_{2}(\hat{\omega}_{1}\hat{a}')_{+}-\hat{\omega}_{1}(\hat{\omega}_{2}\hat{a}')_{-}
    +\hat{\omega}_{1}(\omega_{2}a')_{+}+\hat{\omega}_{2}(\omega_{1}a')_{+}]_{\le n}).
    \end{align}

\end{lem}

\begin{proof}
In order to simplify the calculation, let us check the case $(\omega_{1},0)\star(\omega_{2},0)$. By using Lemma~\ref{thm-dadah}
and the equalities \eqref{pqminus}--\eqref{pqplus}, we have
\begin{align*}
    &(\omega_{1}(z),0)\star(\omega_{2}(z),0) \\
    =&\frac{1}{(2\pi\mathrm{i})^{2}} \oint_{|q-\varphi|>|z-\varphi|}\oint_{|p-\varphi|>|q-\varphi|}\omega_{1}(p)d a(p)\star d a(q)\omega_{2}(q)d p d q\\
    =&\left(\frac{1}{(2\pi\mathrm{i})^{2}}\oint_{|q-\varphi|>|z-\varphi|}\oint_{|p-\varphi|>|q-\varphi|}
    \omega_{1}(p)\left(\frac{a'(p)}{p-q}da(q)+\frac{a'(q)}{q-p}da(p)\right)\omega_{2}(q)d p d q,0\right)\\
    =& \frac{1}{2\pi\mathrm{i}}\left(\oint_{|q-\varphi|>|z-\varphi|}(\omega_{1}(q)a'(q))_{+} \omega_{2}(q)da(q))dq \right. \\
    &\qquad \left.-\oint_{|p-\varphi|>|z-\varphi|}\omega_{1}(p)(\omega_{2}(p)a'(p))_{-}da(p)d p, 0\right)\\
    =&\left( \left( (\omega_{1}(z)a'(z))_{+}\omega_{2}(z)-\omega_{1}(z)(\omega_{2}(z)a'(z))_{-}\right)_{\ge -m+1},0\right). \\
\end{align*}
Observe that the result is indeed symmetric with respect to the indices $1$ and $2$.
The cases $(0,\hat{\omega}_{1})\star(\omega_{2},0)$, $(\omega_{1},0)\star(0,\hat{\omega}_{2})$ and
$(0,\hat{\omega}_{1})\star(0,\hat{\omega}_{2})$ are similar. The lemma is proved.
\end{proof}

As an immediate application of Lemma~\ref{thm-omprod}, we have the following result.
\begin{cor}\label{thm-unitycot}
   For the product defined by \eqref{prodcot}, there is a unity cotangent vector as
\begin{equation}\label{estar}
{\vec{e}}^{\,\ast}=\left(\frac{1}{m}(z-\varphi)^{-m+1},0\right)\in T^{\ast}_{\vec{a}}\mathcal{M}_{m,n}.
\end{equation}
\end{cor}
\begin{proof}
The corollary is verified by using the property $ \left( \frac{(z-\varphi)^{-m+1}}{m} a'(z)\right)_+=1$.
\end{proof}

Summarizing the above discussions, one knows that
$\left(T^{\ast}_{\vec{a}}\mathcal{M}_{m,n}, \star, {\vec{e}}^{\,\ast}, \langle~,\rangle^\ast\right)$ is a Frobenius algebra.
With the help of the bijection $\eta: T^{\ast}_{\vec{a}}\mathcal{M}_{m,n}\to T_{\vec{a}}\mathcal{M}_{m,n}$ defined in
 \eqref{etadef}, it is induced a Frobenius algebra structure on $T_{\vec{a}}\mathcal{M}_{m,n}$. More precisely, we arrive at the following result.
\begin{prop}\label{thm-prodtan}
The quadruple $\left(T_{\vec{a}}\mathcal{M}_{m,n},\circ,\vec{e},\langle~,~\rangle_\eta\right)$ is a Frobenius algebra, in which
the product $\circ$ is given by
\begin{equation}\label{prodtan}
\vec{\xi}_{1} \circ \vec{\xi}_{2}:=\eta\cdot \left(\eta^{-1}(\vec{\xi}_{1})\star\eta^{-1}(\vec{\xi}_{2})\right),\quad \vec{\xi}_{1},\,\vec{\xi}_{2}\in
 T_{\vec{a}}\mathcal{M}_{m,n}
\end{equation}
and the unit vector field  $\vec{e}$ is
\begin{equation}\label{vece}
\vec{e} 
=\left\{
\begin{aligned}
&\frac{\p}{\p t_{0}}+\frac{\p}{\p \hat{h}_{0}},\quad &m&=1,\\
&\frac{1}{m}\frac{\p}{\p h_{m-1}},\quad &m&\ge2.\\
\end{aligned}
\right.
\end{equation}
\end{prop}

\begin{proof}
It suffices to verify the formula \eqref{vece}, say $\vec{e}=\eta\cdot \vec{e}^\ast$. To this end, let us check $\eta^{-1}\left( \vec{e} \right)=\vec{e}^\ast$.  When $m\ge 2$, by using \eqref{flatcov2} one has
\begin{equation}
	\eta^{-1}\left( \frac{1}{m}\frac{\p}{\p h_{m-1}} \right)=\left( \frac{1}{m}(\chi(z)^{-m+1})_{\ge -m+1},0\right)\in T^{\ast}_{\vec{a}}\mathcal{M}_{m,n}.
\end{equation}
This fact together with \eqref{chi}, i.e.,
\begin{align}
	\chi(z)=&z+\chi_{1}z^{-1}+\chi_{2}z^{-2}+\cdots
	=(z-\varphi)+\varphi+\chi_{1}(z-\varphi)^{-1}+\cdots \quad \hbox{near} \quad  \infty,
\end{align}
leads to
\begin{equation}
	\eta^{-1}\left( \frac{1}{m}\frac{\p}{\p h_{m-1}} \right)=\left(\frac{1}{m}(z-\varphi)^{-m+1},0\right)=\vec{e}^\ast.
\end{equation}
When $m=1$, from \eqref{flatcov1} and \eqref{flatcov3} it follows that
\begin{equation}
	\eta^{-1}\left(\frac{\p}{\p t_{0}}+\frac{\p}{\p \hat{h}_{0}}\right)=(1,0)=\vec{e}^\ast.
\end{equation}
The proposition is proved.
\end{proof}

Note that one can also represent the unity vector field in the form $\vec{e}=\frac{\p}{\p v}$ with some flat coordinate $v$ via a linear transformation of $\mathbf{t}\cup\mathbf{h}\cup\hat{\mathbf{h}}$.

\subsection{The symmetric 3-tensor} In this subsection, we want to compute a $3$-tensor $c$ defined by
\begin{equation}\label{3tensor}
c(\p_u,\p_v, \p_w):=\left\la\p_u\circ\p_v, \p_w\right\ra_\eta, \quad u, v,
w\in \mathbf{t}\cup\mathbf{h}\cup\hat{\mathbf{h}},
\end{equation}
where $\p_u=\frac{\p}{\p u}\in T_{\vec{a}}\mathcal{M}_{m,n}$. Clearly, the right hand side of \eqref{3tensor}  is  symmetric with respect to $\p_u$, $\p_v$ and $\p_w$.

\begin{prop}\label{thm-3tensor}
The symmetric $3$-tensor in \eqref{3tensor} can be represented as
\begin{align}
& c(\p_{t_{i_1}},\p_{t_{i_2}},\p_{t_{i_3}})=
\frac{1}{2\pi \mathrm{i}}\oint_{\Gamma}\left(-\zeta^{i_{1}+i_{2}+i_{3}}\zeta'(\zeta'+\zeta'_{-}+\ell') +\{\zeta^{i_{1}
+i_{2}}\zeta'(\zeta^{i_{3}}\zeta')_{-}+\mathrm{c.p.}(i_{1},i_{2},i_{3})\}\right)d z,\nn\\
&c(\p_{h_{j_1}},\p_{h_{j_2}},\p_{h_{j_3}})=-\res_{z=\infty}\left(-\chi^{-j_{1}-j_{2}-j_{3}}\ell'(2\ell'+\zeta'_{-})+
\{\chi^{-j_{1}-j_{2}}\ell'(\ell'\chi^{-j_{3}})_{+}+\mathrm{c.p.}(j_{1},j_{2},j_{3})\}\right)dz,\nn\\
&c(\p_{\hat{h}_{k_1}},\p_{\hat{h}_{k_2}}, \p_{\hat{h}_{k_3}})=\res_{z=\varphi}\left(-\hat{\chi}^{-k_{1}-k_{2}-k_{3}}\ell'(2\ell'-\zeta'_{+})
+ \{\hat{\chi}^{-k_{1}-k_{2}}\ell'(\ell'\hat{\chi}^{-k_{3}})_{-} +\mathrm{c.p.}(k_{1},k_{2},k_{3})\}\right)d z,\nn\\
&c({\p_{t_{i_1}}},{\p_{t_{i_2}}},{\p_{h_{j_3}}})=-\frac{1}{2\pi \mathrm{i}}\oint_{\Gamma}\zeta^{i_{1}+i_{2}}\zeta'(\ell'\chi^{-j_{3}})_{+}dz, \nn\\
&c(\p_{t_{i_1}},\p_{t_{i_2}},\p_{\hat{h}_{k_{3}}})=
\frac{1}{2\pi \mathrm{i}}\oint_{\Gamma}\zeta^{i_{1}+i_{2}}\zeta'(\ell'\hat{\chi}^{-k_{3}})_{-}dz, \nn\\ &c(\p_{h_{j_1}},\p_{h_{j_2}},\p_{t_{i_{3}}})=
\res_{z=\infty}\chi^{-j_{1}-j_{2}}\ell'(\zeta^{i_{3}}\zeta')_{-}dz,\nn \\
&c(\p_{h_{j_1}},\p_{h_{j_2}},\p_{\hat{h}_{k_{3}}})=
\res_{z=\infty}\chi^{-j_{1}-j_{2}}\ell'(\ell'\hat{\chi}^{-k_{3}})_{-}dz,\nn \\ &c(\p_{\hat{h}_{k_1}},\p_{\hat{h}_{k_2}},\p_{t_{i_{3}}})=
-\res_{z=\varphi}\hat{\chi}^{-k_{1}-k_{2}}\ell'(\zeta^{i_{3}}\zeta')_{+}dz,\nn\\
&c(\p_{\hat{h}_{k_1}},\p_{\hat{h}_{k_2}},\p_{h_{j_{3}}})=
-\res_{z=\varphi}\hat{\chi}^{-k_{1}-k_{2}}\ell'(\ell'\chi^{-j_{3}})_{+}dz,\nn\\
&
c(\p_{t_{i_1}},\p_{h_{j_2}},\p_{\hat{h}_{k_{3}}})=0,\nn
\end{align}
where $\chi$ and $\hat{\chi}$ are given in \eqref{chi} and \eqref{chih} respectively.
\end{prop}

Before proving this proposition, let us show a lemma as follows.
\begin{lem}\label{thm-prodflatcov}
The following equalities hold:
\begin{align}
\eta^{-1} \frac{\p}{\p t_{i_{1}}}\star \eta^{-1} \frac{\p}{\p t_{i_{2}}}=&\left([\zeta^{i_{1}+i_{2}}a'
-\zeta^{i_{1}}(\zeta^{i_{2}}\zeta')_{-}-\zeta^{i_{2}}(\zeta^{i_{1}}\zeta')_{-} ]_{\ge -m+1},\right.\nn \\
&\qquad \left. [-\zeta^{i_{1}+i_{2}}\hat{a}'-\zeta^{i_{1}}(\zeta^{i_{2}}\zeta')_{+} -\zeta^{i_{2}}(\zeta^{i_{1}}\zeta')_{+}]_{\le n}\right),
\label{prodflatcov1}\\
\eta^{-1} \frac{\p}{\p h_{j_{1}}}\star \eta^{-1} \frac{\p}{\p h_{j_{2}}}=&\left([-\chi^{-j_{1}-j_{2}}a'+\chi^{-j_{1}}(\chi^{-j_{2}}\ell')_{+}
+\chi^{-j_{2}}(\chi^{-j_{1}}\ell')_{+}]_{\ge -m+1},0\right), \nn\\
\eta^{-1} \frac{\p}{\p \hat{h}_{k_{1}}}\star \eta^{-1} \frac{\p}{\p \hat{h}_{k_{2}}}=&\left(0,[\hat{\chi}^{-k_{1}-k_{2}}\hat{a}'
-\hat{\chi}^{-k_{1}}(\hat{\chi}^{-k_{2}}\ell')_{-} -\hat{\chi}^{-k_{2}}(\hat{\chi}^{-k_{1}}\ell')_{-}]_{\le n}\right),\nn\\
\eta^{-1} \frac{\p}{\p h_{j_{1}}}\star \eta^{-1} \frac{\p}{\p \hat{h}_{k_{2}}}=&\left(-[\chi^{-j_{1}}(\hat{\chi}^{-k_{2}}\ell')_{-}]_{\ge -m+1},
[\hat{\chi}^{-k_{2}}(\chi^{-j_{1}}\ell')_{+}]_{\le n}\right). \label{prodflatcov4}\nn
\end{align}
\end{lem}
\begin{proof}
This lemma is verified by taking Corollary \ref{thm-flatcov} and Lemma \ref{thm-omprod} together, with the help of the
 following equalities for any series $f(z)$ and $g(z)$ in $(z-\vp)$,
    $$
    \begin{aligned}
    &(f_{\ge -m+1}a')_{+}=(f a')_{+},\quad (f_{\ge -m+1}g_{-})_{\ge -m+1}=(f g_{-})_{\ge -m+1},\\
    &(f_{\le n}\hat{a}')_{-}=(f\hat{a}')_{-},\quad (f_{\le n}g_{+})_{\le n}=(f g_{+})_{\le n}.
    \end{aligned}
    $$
For instance, one has
    $$
    \begin{aligned}
    \eta^{-1} \frac{\p}{\p t_{i_{1}}}\star \eta^{-1} \frac{\p}{\p t_{i_{2}}}=((\zeta^{i_{1}})_{\ge -m+1},-(\zeta^{i_{1}})_{\le n})\star
    ((\zeta^{i_{2}})_{\ge -m+1},-(\zeta^{i_{2}})_{\le n})=(\om,\hat{\om}),\\
    \end{aligned}
    $$
where
\begin{align*}
\om=&[(\zeta^{i_{2}})_{\ge -m+1}(\zeta^{i_{1}}a')_{+}-\zeta^{i_{1}}((\zeta^{i_{2}})_{\ge -m+1}a')_{-} \\
&\qquad +(\zeta^{i_{2}})_{\ge-m+1}(\zeta^{i_{1}}\hat{a}')_{-}+(\zeta^{i_{1}})_{\ge-m+1}(\zeta^{i_{2}}\hat{a}')_{-}]_{\ge -m+1}\\
=&[-(\zeta^{i_{2}})_{\ge -m+1}(\zeta^{i_{1}}a')_{-}+(\zeta^{i_{2}})_{\ge -m+1}\zeta^{i_{1}}a'-\zeta^{i_{1}}((\zeta^{i_{2}})_{\ge -m+1}a')_{-} \\
&\qquad +\zeta^{i_{2}}(\zeta^{i_{1}}\hat{a}')_{-} +\zeta^{i_{1}}(\zeta^{i_{2}}\hat{a}')_{-}]_{\ge -m+1}\\
=&[-\zeta^{i_{2}}(\zeta^{i_{1}}a')_{-}+\zeta^{i_{1}}(\zeta^{i_{2}}a')_{+}+\zeta^{i_{2}}(\zeta^{i_{1}}\hat{a}')_{-}
+\zeta^{i_{1}}(\zeta^{i_{2}}\hat{a}')_{-}]_{\ge -m+1}\\
    =&[\zeta^{i_{1}}\zeta^{i_{2}}a'-\zeta^{i_{1}}(\zeta^{i_{2}}(a'-\hat{a}'))_{-} -\zeta^{i_{2}}(\zeta^{i_{1}}(a'-\hat{a}'))_{-} ]_{\ge -m+1}\\
    =&[\zeta^{i_{1}+i_{2}}a'-\zeta^{i_{1}}(\zeta^{i_{2}}\zeta')_{-}-\zeta^{i_{2}}(\zeta^{i_{1}}\zeta')_{-}]_{\ge -m+1}, \\
\hat{\om}=&[\zeta^{i_{2}}((\zeta^{i_{1}})_{\le n}\hat{a}')_{+}-(\zeta^{i_{1}})_{\le n}(\zeta^{i_{2}}\hat{a}')_{-}
-\zeta^{i_{2}}((\zeta^{i_{1}})_{\ge-m+1}a')_{+}-\zeta^{i_{1}}((\zeta^{i_{2}})_{\ge-m+1}a')_{+}]_{\le n} \\
=&[-\zeta^{i_{2}}((\zeta^{i_{1}})_{\le n}\hat{a}')_{-}+\zeta^{i_{2}}(\zeta^{i_{1}})_{\le n}\hat{a}'
-(\zeta^{i_{1}})_{\le n}(\zeta^{i_{2}}\hat{a}')_{-}  -\zeta^{i_{2}}(\zeta^{i_{1}}a')_{+}-\zeta^{i_{1}}(\zeta^{i_{2}}a')_{+}]_{\le n} \\
=&[-\zeta^{i_{2}}(\zeta^{i_{1}}\hat{a}')_{-}+\zeta^{i_{1}}(\zeta^{i_{2}}\hat{a}')_{+}
-\zeta^{i_{2}}(\zeta^{i_{1}}a')_{+}-\zeta^{i_{1}}(\zeta^{i_{2}}a')_{+}]_{\le n}\\
    =&[-\zeta^{i_{2}}\zeta^{i_{1}}\hat{a}'+\zeta^{i_{1}}(\zeta^{i_{2}}(\hat{a}'-a'))_{+} +\zeta^{i_{2}}(\zeta^{i_{1}}(\hat{a}'-a'))_{+} ]_{\le n}\\
    =&[-\zeta^{i_{1}+i_{2}}\hat{a}'-\zeta^{i_{1}}(\zeta^{i_{2}}\zeta')_{+}-\zeta^{i_{2}}(\zeta^{i_{1}}\zeta')_{+}]_{\le n}.
\end{align*}
So the first equality in \eqref{prodflatcov1} is valid. The other cases are similar, thus the lemma is proved.
\end{proof}

\begin{prfof}{Proposition~\ref{thm-3tensor}.}
According to \eqref{pairing}, \eqref{metric} and \eqref{prodtan}, we have
\begin{equation}\label{3tensor2}
\left\la\p_u\circ\p_v, \p_w\right\ra_\eta= \left\langle \eta^{-1}\p_u\star \eta^{-1}\p_v, \p_w\right\rangle, \quad u, v,
w\in  \mathbf{t}\cup\mathbf{h}\cup\hat{\mathbf{h}}.
\end{equation}
Let us substitute into it with the data in Lemma~\ref{thm-prodflatcov} and Corollary~\ref{thm-flatv}. For instance,
\begin{align*}
    &c(\p_{t_{i_1}},\p_{t_{i_2}},\p_{t_{i_3}})= \left\langle \frac{\p}{\p t_{i_{1}}}\circ \frac{\p}{\p t_{i_{2}}},
    \frac{\p}{\p t_{i_{3}}}\right\rangle_{\eta} \\
    =&\frac{1}{2\pi\mathrm{i}}\oint_{\Gamma}
    \left([\zeta^{i_{1}+i_{2}}a'-\zeta^{i_{1}}(\zeta^{i_{2}}\zeta')_{-}-\zeta^{i_{2}}(\zeta^{i_{1}}\zeta')_{-} ]
     (-\zeta^{i_3}\zeta')_{-} \right.\nn \\
&\qquad \left. +[-\zeta^{i_{1}+i_{2}}\hat{a}'-\zeta^{i_{1}}(\zeta^{i_{2}}\zeta')_{+} -\zeta^{i_{2}}(\zeta^{i_{1}}\zeta')_{+}]
(\zeta^{i_3}\zeta')_{+}\right) d z\\
 =&\frac{1}{2\pi\mathrm{i}}\oint_{\Gamma}
    \left([-\zeta^{i_{1}+i_{2}}a'+\zeta^{i_{1}}(\zeta^{i_{2}}\zeta')_{-}+\zeta^{i_{2}}(\zeta^{i_{1}}\zeta')_{-} ]
    \zeta^{i_3}\zeta'  \right.\nn \\
&\qquad \left. +[-\zeta^{i_{1}+i_{2}}(\hat{a}'-a')-\zeta^{i_{1}}\zeta^{i_{2}}\zeta'  -\zeta^{i_{2}}\zeta^{i_{1}}\zeta' ]
 (\zeta^{i_3}\zeta')_{+}\right) d z \\
 =&\frac{1}{2\pi\mathrm{i}}\oint_{\Gamma}
    \left(-\zeta^{i_{1}+i_{2}+i_3}a'+\zeta^{i_{1}+i_3}\zeta'(\zeta^{i_{2}}\zeta')_{-}
    +\zeta^{i_{2}+i_3}\zeta'(\zeta^{i_{1}}\zeta')_{-} -\zeta^{i_{1}+i_2}\zeta'(\zeta^{i_3}\zeta')_{+}\right) d z \\
 =&\frac{1}{2\pi\mathrm{i}}\oint_{\Gamma}
    \left(-\zeta^{i_{1}+i_{2}+i_3}(\zeta'+a')+\left\{\zeta^{i_{1}+i_3}\zeta'(\zeta^{i_{2}}\zeta')_{-}
    +\mathrm{c.p.}(i_1,i_2,i_3)\right\} \right) d z.
\end{align*}
The other cases are similar. Thus the proposition is proved.
\end{prfof}

\subsection{The potential function}

We proceed to show that the values $<\p_u\circ\p_v, \p_w>_{\eta}$ in Proposition~\ref{thm-3tensor} can be represented as the third-order
derivatives of a certain potential function. To this end, let us state the following lemma first.

\begin{lem} \label{thm-V1V2}
There locally exit three functions $V_1\left(\mathbf{t}\right)$, $V_2\left( \mathbf{h}\right)$ and  $G\left(\mathbf{h},
\hat{\mathbf{h}}\right)$ such that
\begin{align}\label{}
&\frac{\p^{2}V_{1}}{\p t_{i_1}\p t_{i_2}}=\frac{1}{2\pi \mathrm{i}}\oint_\Gamma\frac{\zeta'(z)\zeta(z)^{i_1+i_2}}{z}d z, \label{V1} \\
&\frac{\p^{2}V_{2}}{\p h_{j_1}\p h_{j_2}}=-\res_{z=\infty}\frac{\ell'(z)\chi(z)^{-j_1-j_2}}{z}d z, \label{V2} \\
&\frac{\p^3 G}{\p u \p v \p w}=-\left(\res_{z=\infty}+\res_{z=\varphi}\right)\frac{\p_{\mathit{u}}\ell(z)\cdot
 \p_{\mathit{v}}\ell(z)\cdot \p_{\mathit{w}}\ell(z)}{\ell'(z)}d z  \label{Suvw}
\end{align}
with $u,v,w\in\mathbf{h}\cup\hat{\mathbf{h}}$.

\end{lem}
\begin{proof}
Firstly, it follows from \eqref{zetaellt} that
\begin{align*}
&\frac{1}{2\pi \mathrm{i}}\frac{\p}{\p t_{i_3}} \oint_\Gamma\frac{ \zeta'(z)\zeta(z)^{i_1+i_2} }{z}d z \\
=& \frac{1}{2\pi\mathrm{i}}\oint_\Gamma\frac{-{(\zeta'(z)\zeta(z)^{i_3})'\cdot\zeta(z)^{i_1+i_2} - \zeta'(z)\cdot (i_1+i_2)\zeta'(z)\zeta(z)^{i_1+i_2+i_3-1} }(z)}{z}d z \\
=& \frac{1}{2\pi\mathrm{i}}\oint_\Gamma\frac{-\left( \zeta'(z)\zeta(z)^{i_1+i_2+i_3} \right)' }{z}d z,
\end{align*}
which is symmetric with respect to the indices $i_1$, $i_2$ and $i_3$. So the existence of $V_1$ is verified.

Secondly, in consideration of the form of $\ell(z)$ and $\chi(z)$, we observe that the right hand side of \eqref{V2}
depends only on $\{h_j\}_{j=1}^{m-1}$. In the same way as above, by using \eqref{zetaellh} we have
\begin{align*}
&-\frac{\p   }{\p h_{j_3}}\res_{z=\infty}\frac{\ell'(z)\chi(z)^{-j_1-j_2}}{z}d z \\
=& -\res_{z=\infty}\left( \left(\ell'(z)\chi(z)^{-j_3}\right)'_+\chi(z)^{-j_1-j_2}
-\ell'(z)\cdot\frac{j_1+j_2}{m}\chi(z)^{-j_1-j_2-m}\left(\ell'(z)\chi(z)^{-j_3}\right)_+\right)\frac{d z}{z}  \\
=& -\res_{z=\infty}\left( \left(\ell'(z)\chi(z)^{-j_3}\right)' \chi(z)^{-j_1-j_2}
+\left(\chi(z)^{-j_1-j_2}(z)\right)' \ell'(z)\chi(z)^{-j_3} \right)\frac{d z}{z} \\
=& -\res_{z=\infty} \left(\ell'(z)\chi(z)^{-j_1-j_2-j_3}\right)' \frac{d z}{z},
\end{align*}
which is symmetric with respect to the indices $j_1$, $j_2$ and $j_3$. Hence the function $V_2$ exists.

The existence of $G$ has been proved by Krichever in \cite{Kr} (see Theorem~5.6 therein for a more general conclusion). For convenience of the readers, let us verify it in brief now.
To this end, denote the right hand side of \eqref{Suvw} as $G_{\mathit{u},\mathit{v},\mathit{w}}$, then it is sufficient
to show  that $ \p_s G_{\mathit{u},\mathit{v},\mathit{w}}$ is symmetric with respect to $\mathit{u},\mathit{v},
\mathit{w},\mathit{s}\in \mathbf{h}\cup\hat{\mathbf{h}}$.
For instance, we have
\begin{align}
&G_{h_{j_{1}},h_{j_{2}},h_{j_{3}}} \nn \\
=&-\res_{z=\infty}\frac{\p_{h_{j_{1}}}\ell\cdot \p_{h_{j_{2}}}\ell\cdot \p_{h_{j_{3}}}\ell}{\ell'}dz\nn\\
        =&-\res_{z=\infty}\frac{(\ell'\chi^{-j_{1}})_{+} (\ell'\chi^{-j_{2}})_{+}(\ell'\chi^{-j_{3}})_{+}}{\ell'}d z\nn\\
        =&-\res_{z=\infty}\frac{[\ell'\chi^{-j_{1}} -(\ell'\chi^{-j_{1}})_{-}][\ell'\chi^{-j_{2}} -(\ell'\chi^{-j_{2}})_{-}][\ell'\chi^{-j_{3}}-(\ell'\chi^{-j_{3}})_{-}]}{\ell'}d z\nn\\
        =&-\res_{z=\infty}\left(\ell'\ell'\chi^{-j_{1}-j_{2}-j_{3}} -\{\ell'\chi^{-j_{1}-j_{2}}(\ell'\chi^{-j_{3}})_{-}+\mathrm{c.p.}(j_{1},j_{2},j_{3})\}\right) d z \nn\\
         =&-\res_{z=\infty}\left(R_{j_{1}+j_{2}}(R_{j_{3}})_{+} -R_{j_{1}+j_{3}}(R_{j_{2}})_{-}-R_{j_{2}+j_{3}}(R_{j_{1}})_{-}\right)dz\nn\\
         =&-\res_{z=\infty}\left(R_{j_{1}+j_{2}}(R_{j_{3}})_{+} -(R_{j_{1}+j_{3}})_+R_{j_{2}}-(R_{j_{2}+j_{3}})_+R_{j_{1}} \right)dz, \label{S3h}
    \end{align}
where $R_{j}(z)=\ell'(z)\chi(z)^{-j}$. With the help of \eqref{zetaellh}, it is easy to see that
\begin{equation}\label{Rj}
R_{j_1}R_{j_2+j_3}=R_{j_2}R_{j_3+j_1}, \quad \frac{\p R_{j_1}}{\p h_{j_2}}=\left(\chi^{-j_1}(R_{j_2})_{+}\right)', \quad  \left(\frac{\p R_{j_1}}{\p h_{j_2}}\right)_{+}=(R_{j_1+j_2})_{+}'.
\end{equation}
In order to avoid lengthy notations, let us write $f\backsim g$ if $f-g$ is symmetric with respect to the indices $j_{3}$ and $j_{4}$, then it is straightforward to calculate
    \begin{align*}
\frac{G_{h_{j_{1}},h_{j_{2}},h_{j_{3}}} }{ \p{h_{j_{4}}} } \backsim&-\res_{z=\infty}\left( (\chi^{-j_{1}-j_{2}}(R_{j_{4}})_{+})'(R_{j_{3}})_{+} -(R_{j_{1}+j_{3}})_+(\chi^{-j_{2}}(R_{j_{4}})_{+})' \right.\\
&\qquad\left.-(R_{j_{2}+j_{3}})_+(\chi^{-j_{1}}(R_{j_{4}})_{+})'\right)dz\\
\backsim&-\res_{z=\infty}\left(\chi^{-j_{1}}(\chi^{-j_{2}}(R_{j_{4}})_{+})'(R_{j_{3}})_{+} -(\chi^{-j_1}(R_{j_{3}})_+)_+(\chi^{-j_{2}}(R_{j_{4}})_{+})'_- \right.\\
&\qquad\left.-(\chi^{-j_2}(R_{j_{3}})_+)_+(\chi^{-j_{1}}(R_{j_{4}})_{+})'\right)dz\\
\backsim&-\res_{z=\infty}\left(\chi^{-j_{1}}(\chi^{-j_{2}}(R_{j_{4}})_{+})'(R_{j_{3}})_{+} -\chi^{-j_1}(R_{j_{3}})_+(\chi^{-j_{2}}(R_{j_{4}})_{+})'_- \right.\\
&\qquad \left.-(\chi^{-j_2}(R_{j_{3}})_+)_+(\chi^{-j_{1}}(R_{j_{4}})_{+})'\right)dz\\
\backsim&-\res_{z=\infty}\left(\chi^{-j_{1}}(R_{j_{3}})_{+}(\chi^{-j_{2}}(R_{j_{4}})_{+})_+'  +(\chi^{-j_2}(R_{j_{3}})_+)_+'\chi^{-j_{1}}(R_{j_{4}})_{+}\right)dz\\
 \backsim&\quad0,
    \end{align*}
which leads to
\[
    \frac{ G_{h_{j_{1}},h_{j_{2}},h_{j_{3}}} }{\p {h_{j_{4}}}}=\frac{ G_{h_{j_{1}},h_{j_{2}},h_{j_{4}}} }{\p {h_{j_{3}}}}.
\]
Furthermore, from \eqref{zetaellhh} it follows that
\[
\p_{\hat{h}_k}\ell(z)=\p_{\hat{h}_k}\ell(z)_-, \quad \left(\frac{\p_{\hat{h}_{k}}\ell_-}{\ell'}\right)_-=0 \hbox{ as } |z-\vp|\to0.
\]
The first equality together with  $\p_{h_j}\ell(z)=\p_{h_j}\ell(z)_+$ implies  $\p_{h_{j}}\p_{\hat{h}_{k}}\ell=0$, hence we have
    \begin{align}
    G_{h_{j_{1}},h_{j_{2}},\hat{h}_{k}}  =&-(\res_{z=\infty}+\res_{z=\vp})\frac{\p_{h_{j_{1}}}\ell_+\cdot \p_{h_{j_{2}}}\ell_+\cdot \p_{\hat{h}_{k}}\ell_-}{\ell'}dz \nn\\
    =&-\res_{z=\infty}\frac{[\ell'\chi^{-j_{1}}-(\ell'\chi^{-j_{1}})_{-}][\ell'\chi^{-j_{2}} -(\ell'\chi^{-j_{2}})_{-}]\p_{\hat{h}_{k}}\ell_-}{\ell'}dz \nn\\
    =&-\res_{z=\infty}(\ell'\chi^{-j_{1}-j_{2}})_{+}\p_{\hat{h}_{k}}\ell_- dz, \label{S3h2}
    \end{align}
From the first and the second equalities in \eqref{S3h} it follows that
\begin{align*}
&\frac{\p G_{h_{j_{1}},h_{j_{2}},h_{j_{3}}}}{\p{\hat{h}_{k}}} \\ =&\res_{z=\infty}\frac{\p_{h_{j_{1}}}\ell_+\cdot \p_{h_{j_{2}}}\ell_+\cdot \p_{h_{j_{3}}}\ell_+\cdot \p_{\hat{h}_{k}}\ell'_-}{\ell'\cdot \ell'}dz\\
    =&\res_{z=\infty}\frac{[\ell'\chi^{-j_{1}}-(\ell'\chi^{-j_{1}})_{-}][\ell'\chi^{-j_{2}} -(\ell'\chi^{-j_{2}})_{-}][\ell'\chi^{-j_{3}} -(\ell'\chi^{-j_{3}})_{-}]\p_{\hat{h}_{k}}\ell'_-}{\ell'\cdot \ell'}dz\\
    =&\res_{z=\infty}\ell'\chi^{-j_{1}-j_{2}-j_{3}}(\p_{\hat{h}_{k}}\ell_-)' dz
   \\
    =&-\res_{z=\infty}(\ell'\chi^{-j_{1}-j_{2}-j_{3}})'_+\p_{\hat{h}_{k}}\ell_- dz
    \\
    =&-\res_{z=\infty}\frac{\p(\ell'\chi^{-j_{1}-j_{2}})_{+} }{\p {h_{j_{3}}}} \p_{\hat{h}_{k}}\ell_- dz \\
    =&\frac{\p G_{h_{j_{1}},h_{j_{2}},\hat{h}_{k}} }{\p {h_{j_{3}}}},
\end{align*}
where the last equality is due to \eqref{S3h2} and \eqref{Rj}.
    The other cases are similar. Therefore the lemma is proved.
\end{proof}

\begin{rem}\label{rmk-V1V2S}
Each function of $V_1$ and $V_2$ is determined up to a linear function of its arguments, while the function $G$ is determined
 up to a quadratic function of its arguments. One refers to, for example \cite{AK-CMP} for some concrete examples of $G$.
\end{rem}

Under the assumptions in \eqref{GmGm},
for any $z_1\in\Gamma'$, $z_2\in\Gamma$, $\vp\in U$,  one has
\begin{align*}
&\log\left(\frac{z_{1}-z_{2}}{z_{1}}\right)=-\sum_{i\ge1}\frac{1}{i}\left(\frac{z_{2}}{z_{1}}\right)^{i}, \\
&\frac{1}{z_{1}-z_{2}}=\frac{1}{(z_{1}-\vp)-(z_{2}-\vp)} =\sum_{i\ge0}\frac{(z_{2}-\vp)^i}{(z_{1}-\vp)^{i+1}}.
\end{align*}
Let $f(z)=\sum_{-\infty}^{+\infty}f_{i}(z-\varphi)^{i}$ be an arbitrary holomorphic function in a neighborhood of the closed bend with
boundaries $\Gamma$ and $\Gamma'$, then it is easy to verify the following equalities:
\begin{align}
\frac{1}{2\pi \mathrm{i}}\oint_{\Gamma'}f'(z_{2})\log\left(\frac{z_{1}-z_{2}}{z_{1}}\right)dz_{2}=
&\frac{1}{2\pi \mathrm{i}}\oint_{\Gamma'}\frac{f(z_{2})}{z_{1}-z_{2}}dz_{2}=f(z_{1})_{-},  \label{intlog1} \\
\frac{1}{2\pi \mathrm{i}}\oint_{\Gamma}f'(z_{1})\log\left(\frac{z_{1}-z_{2}}{z_{1}}\right)dz_{1}=
&\frac{1}{2\pi \mathrm{i}}\oint_{\Gamma}f(z_{1})\left(\frac{-1}{z_{1}-z_{2}}+\frac{1}{z_1}\right)d z_{1}
\nn\\
=&-f(z_{2})_{+}+\frac{1}{2\pi \mathrm{i}}\oint_{\Gamma}\frac{f(z)}{z}d z. \label{intlog2}
\end{align}

Now we are ready to introduce a function of $\mathbf{t}\cup\mathbf{h}\cup\hat{\mathbf{h}}$ as
\begin{align}
\mathcal{F}=&\frac{1}{(2\pi \mathrm{i})^{2}}\oint_{\Gamma}\oint_{\Gamma'}\left(\frac{1}{2}\zeta(z_{1}) \zeta(z_{2})
+\zeta(z_{1}) \ell(z_{2}) -\ell(z_{1}) \zeta(z_{2})\right) \log\left(\frac{z_{1}-z_{2}}{z_{1}}\right)dz_{1}dz_{2} \nn\\
&\qquad -\frac{V_1(\mathbf{t})}{2\pi \mathrm{i} }\oint_\Gamma \left(\frac{1}{2}\zeta(z)+\ell(z)\right) d z
+ \frac{V_2(\mathbf{h})}{2\pi \mathrm{i} }\oint_\Gamma \zeta(z) d z +G(\mathbf{h},\hat{\mathbf{h}}). \label{potential}
\end{align}
In fact, from the equalities \eqref{ti} and \eqref{chih} it follows that
\begin{align}\label{}
t_{-1}= &\frac{1}{2\pi \mathrm{i} }\oint_\Sigma z(\zeta)d\zeta = -\frac{1}{2\pi \mathrm{i} }\oint_\Gamma \zeta(z) d z,  \\
n\hat{h}_n=&-n\res_{\hat{\chi}=\infty}z(\hat{\chi})\hat{\chi}^{n-1}d \hat{\chi} =\res_{z=\vp}\hat{\chi}^{n}(z)d z
= \frac{1}{2\pi \mathrm{i} }\oint_\Gamma \ell(z) d z,
\end{align}
hence the function $\mathcal{F}$ can be represented equivalently as \eqref{potential0}. Clearly, by virtue of
 Remark~\ref{rmk-V1V2S}, the function $\mathcal{F}$ is determined up to a quadratic function of the flat coordinates.

\begin{prop}\label{thm-potential}
    The function $\mathcal{F}$ given above satisfies, for any $u, v, w\in \mathbf{t}\cup\mathbf{h}\cup\hat{\mathbf{h}}$,
\begin{equation}\label{Fuvw}
\frac{\p^{3} \mathcal{F}}{\p \mathit{u} \p \mathit{v}\p \mathit{w}}=\left\la\p_u\circ\p_v, \p_w\right\ra_\eta.
\end{equation}
\end{prop}

\begin{proof}
Let us verify the equalities \eqref{Fuvw} in a straightforward way, with the help of \eqref{zetaellt}--\eqref{zetaellhh},
\eqref{intlog1} and \eqref{intlog2}.
For convenience, denote $Z_i(z)=\zeta(z)^i\zeta'(z)$, then from \eqref{zetaellt} it is easy to see that
\begin{equation}\label{}
\frac{\p\zeta(z)}{\p t_i}=-Z_i(z), \quad 
 \frac{\p^2\zeta(z)}{\p t_{i_1}\p t_{i_2}}={Z_{i_1+i_2}}'(z), \quad  \frac{\p^3\zeta(z)}{\p t_{i_1}\p t_{i_2}\p t_{i_3}}=-{Z_{i_1+i_2+i_3}}''(z).
\end{equation}
By using these equalities and \eqref{V1}, we have
\begin{align*}
&\frac{\p^{3} \mathcal{F}}{\p t_{i_{1}} \p t_{i_{2}}\p t_{i_{3}}} \\
=&\frac{1}{(2\pi \mathrm{i})^{2}}\oint_\Gamma\oint_{\Gamma'}\left( (-{Z_{i_1+i_2+i_3}}''(z_1)\left(\frac{1}{2}\zeta(z_2)+\ell(z_2)\right)
- \left(\frac{1}{2}\zeta(z_1)-\ell(z_1)\right){Z_{i_1+i_2+i_3}}''(z_2) \right. \\
&\left.-\frac{1}{2}\left\{{Z_{i_1+i_2}}'(z_1)Z_{i_3}(z_2) +Z_{i_3}(z_1){Z_{i_1+i_2}}'(z_2) +\mathrm{c.p.}(i_{1},i_{2},i_{3})\right\} \right)
\log\left(\frac{z_{1}-z_{2}}{z_{1}}\right)  dz_{1}dz_{2}\\
&-\frac{1}{2\pi \mathrm{i} }\oint_\Gamma \left(\frac{1}{2}\zeta(z)
+\ell(z)\right) d z \cdot \frac{1}{2\pi \mathrm{i}}\oint_\Gamma\frac{-{Z_{i_1+i_2+i_3}}'(z)}{z}d z \\
&-\left\{\frac{1}{2\pi \mathrm{i} }\oint_\Gamma \left(-\frac{1}{2}Z_{i_3}(z)\right) d z \cdot \frac{1}{2\pi \mathrm{i}}
\oint_\Gamma\frac{{Z_{i_1+i_2}}(z)}{z}d z +\mathrm{c.p.}(i_1,i_2,i_3)\right\}
\\
=&\frac{1}{2\pi \mathrm{i}}\oint_\Gamma \left(  {Z_{i_1+i_2+i_3}}'(z)_+ \left(\frac{1}{2}\zeta(z)+\ell(z)\right)
- \left(\frac{1}{2}\zeta(z)-\ell(z)\right){Z_{i_1+i_2+i_3}}'(z)_- \right. \\
&\left.-\frac{1}{2}\left\{-{Z_{i_1+i_2}}(z)_+Z_{i_3}(z) +Z_{i_3}(z){Z_{i_1+i_2}}(z)_- +\mathrm{c.p.}(i_{1},i_{2},i_{3})\right\} \right)
  d z \\
&-\frac{1}{2\pi \mathrm{i}}\oint_{\Gamma'} \frac{{Z_{i_1+i_2+i_3} }'(z)}{z}d z  \cdot \frac{1}{2\pi \mathrm{i}}
\oint_\Gamma \left(\frac{1}{2}\zeta(z)+\ell(z)\right)d z  \\
&-\frac{1}{2}\left\{ \frac{1}{2\pi \mathrm{i}}\oint_{\Gamma'} \frac{{Z_{i_1+i_2}}(z)}{z}d z
\cdot \frac{1}{2\pi \mathrm{i}}\oint_\Gamma Z_{i_3}(z)d z+\mathrm{c.p.}(i_{1},i_{2},i_{3})\right\}  \\
&-\frac{1}{2\pi \mathrm{i} }\oint_\Gamma \left(\frac{1}{2}\zeta(z)
+\ell(z)\right) d z \cdot \frac{1}{2\pi \mathrm{i}}\oint_\Gamma\frac{-{Z_{i_1+i_2+i_3}}'(z)}{z}d z \\
&-\left\{\frac{1}{2\pi \mathrm{i} }\oint_\Gamma \left(-\frac{1}{2}Z_{i_3}(z)\right) d z \cdot \frac{1}{2\pi \mathrm{i}}
\oint_\Gamma\frac{{Z_{i_1+i_2}}(z)}{z}d z +\mathrm{c.p.}(i_1,i_2,i_3)\right\}
\\
=&\frac{1}{2\pi \mathrm{i}}\oint_\Gamma \left(  {Z_{i_1+i_2+i_3}}(z) \left( -\frac{1}{2}\zeta'(z)_- -\ell'(z)_-
 + \frac{1}{2}\zeta'(z)_+ -\ell'(z)_+ \right)  \right. \\
&\left. - \frac{3}{2} Z_{i_3}(z){Z_{i_1+i_2}}(z) + \left\{{Z_{i_1+i_2}}(z)_+Z_{i_3}(z)  +\mathrm{c.p.}(i_{1},i_{2},i_{3})\right\} \right)
  d z +0
\\
=&\frac{1}{2\pi \mathrm{i}}\oint_\Gamma \left(  {Z_{i_1+i_2+i_3}}(z) \left(-\zeta'(z) - \zeta'(z)_- -\ell'(z) \right)
+ \left\{{Z_{i_1+i_2}}(z) Z_{i_3}(z)_-  +\mathrm{c.p.}(i_{1},i_{2},i_{3})\right\} \right)
  d z \\
=&\left\langle \frac{\p}{\p t_{i_{1}}}\circ \frac{\p}{\p t_{i_{2}}},\frac{\p}{\p t_{i_{3}}}\right\rangle_{\eta},
\end{align*}
where the second equality is due to \eqref{intlog1} and \eqref{intlog2}, and the fourth equality is due to
$Z_{i_3}(z){Z_{i_1+i_2}}(z)={Z_{i_1+i_2+i_3}}(z)\zeta'(z)$.

Similarly, by using \eqref{zetaellh} and \eqref{V2} we have
\begin{align*}
&\frac{\p^{3} \mathcal{F}}{\p t_{i_{1}} \p t_{i_{2}}\p h_{j}}
\\
=&\frac{1}{(2\pi \mathrm{i})^{2}}\oint_\Gamma\oint_{\Gamma'}\left( {Z_{i_1+i_2}}'(z_1) (\ell'(z_2)\chi(z_2)^{-j})_+\right.  \\
&\left.- (\ell'(z_1)\chi(z_1)^{-j})_+{Z_{i_1+i_2}}'(z_2)  \right)
\log\left(\frac{z_{1}-z_{2}}{z_{1}}\right)  dz_{1}dz_{2}
\\
=&-\frac{1}{2\pi \mathrm{i}}\oint_\Gamma \bigg( \left( {Z_{i_1+i_2}}(z_1)_+
-  \frac{1}{2\pi \mathrm{i}}\oint_{\Gamma'}\frac{{Z_{i_1+i_2}}(z_1)}{z_1}d z_1\right) (\ell'(z_2)\chi(z_2)^{-j})_+   \\
& - (\ell'(z)\chi(z)^{-j})_+{Z_{i_1+i_2}}(z)_-  \bigg) d z
\\
=&-\frac{1}{2\pi \mathrm{i}}\oint_\Gamma  (\ell'(z)\chi(z)^{-j})_+\zeta(z)^{i_1+i_2}\zeta'(z)  d z
\\
=&\left\langle \frac{\p}{\p t_{i_{1}}}\circ \frac{\p}{\p t_{i_{2}}},\frac{\p}{\p h_{j}}\right\rangle_{\eta}.
\end{align*}
Moreover, by using  \eqref{Rj} and \eqref{S3h}  we have
\begin{align*}
&\frac{\p^{3} \mathcal{F}}{\p h_{j_{1}} \p h_{j_{2}}\p h_{j_3}}
\\
=&\frac{1}{(2\pi \mathrm{i})^{2}}\oint_\Gamma\oint_{\Gamma'}\left( \zeta(z_1) (R_{j_1+j_2+j_3}(z_2))''_+
- (R_{j_1+j_2+j_3}(z_1))''_+ \zeta(z_2)  \right)
\log\left(\frac{z_{1}-z_{2}}{z_{1}}\right)  dz_{1}dz_{2} \\
&+\frac{1}{2\pi \mathrm{i} }\oint_\Gamma \zeta(z) d z\cdot \frac{1}{2\pi \mathrm{i} }\frac{\p}{\p h_{j_3}}
\oint_\Gamma \frac{R_{j_1+j_2}(z)_+}{z} d z  + \frac{\p^{3} G}{\p h_{j_{1}} \p h_{j_{2}}\p h_{j_3}}
\\
=&\frac{1}{2\pi \mathrm{i}}\oint_\Gamma  {R_{j_1+j_2+j_3}}'(z)_+ \zeta(z) d z - \frac{1}{2\pi \mathrm{i} }
 \oint_{\Gamma'} \frac{ {R_{j_1+j_2+j_3}}'(z)_+ }{z} d z \cdot \frac{1}{2\pi \mathrm{i} }\oint_\Gamma \zeta(z) d z  \\
&+\frac{1}{2\pi \mathrm{i} }\oint_\Gamma \zeta(z) d z\cdot \frac{1}{2\pi \mathrm{i} } \oint_\Gamma \frac{-{R_{j_1+j_2+j_3}}'(z)_+}{z} d z
\\
&+\res_{z=\infty}   \left(2 R_{j_1+j_2+j_3}(z) \ell'(z)  -\{ R_{j_1+j_2}(z) R_{j_{3}}(z)_{+}+\mathrm{c.p.}(j_{1},j_{2},j_{3})\}\right) d z
\\
=&\res_{z=\infty}   \left( R_{j_1+j_2+j_3}(z) (\zeta'(z)_- +2\ell'(z))  -\{ R_{j_1+j_2}(z) R_{j_{3}}(z)_{+}
+\mathrm{c.p.}(j_{1},j_{2},j_{3})\}\right) d z
\\
=&\left\langle \frac{\p}{\p h_{j_1} }\circ \frac{\p}{\p h_{j_2}},\frac{\p}{\p h_{j_3}}\right\rangle_{\eta},
\end{align*}
The other cases are similar. Therefore the proposition is proved.
\end{proof}

\subsection{The Euler vector field}

On $\mathcal{M}_{m,n}$, let us introduce a vector field as
\begin{equation}\label{vecE}
      \vec{E}=\sum_{i\in\mathbb{Z}}\left(\frac{1}{m}-i\right)t_{i}\frac{\p}{\p t_{i}}+\sum_{j=1}^{m-1} \frac{j+1}{m}  h_{j}\frac{\p}{\p h_{j}}+
\sum_{k=0}^{n}\left(\frac{1}{m}+\frac{k}{n}\right)\hat{h}_{k}\frac{\p}{\p \hat{h}_{k}}.
\end{equation}
\begin{prop}\label{thm-EF}
The function $\mathcal{F}$ defined by \eqref{potential} satisfies
\begin{equation}\label{EF}
\mathrm{Lie}_{\vec{E}} \mathcal{F}=\left(2+\frac{2}{m}\right)\mathcal{F}+\hbox{quadratic terms in flat coordinates}.
\end{equation}
\end{prop}
\begin{proof}
If we assign $\deg\,z=\frac{1}{m}$, and assume each of the functions $\zeta(z)$ and $\ell(z)$ to be homogeneous of degree $1$,
then from \eqref{zetaellt}--\eqref{zetaellhh} one sees that the flat coordinates have degrees as follows
 \begin{align}\label{degree}
   &\deg\,t_i=\frac{1}{m}-i,\quad i\in\Z; \\
   &  \deg\,h_j=\frac{j+1}{m},\quad
   1\le j\le m-1; \\
   &     \deg\,\hat{h}_k=\frac{k}{n}+\frac{1}{m}, \quad 0\le k\le n. \label{degree3}
 \end{align}
Accordingly, one observes that $V_1$ and $V_2$ defined by \eqref{V1} and \eqref{V2} are homogeneous of degree $1+\frac{1}{m}$
(up to a linear terms in the flat coordinates)
while $G$ defined by \eqref{Suvw} is homogeneous of degree $2+\frac{2}{m}$ (up to a quadratic terms in the flat coordinates).
Hence, up to quadratic terms the function $\mathcal{F}$ is homogeneous of degree
\[
\deg\,\mathcal{F}=2+\frac{2}{m}.
\]
Therefore the proposition is proved.
\end{proof}

\subsection{Proof of Main Theorem 1} It follows from
 Propositions~\ref{thm-flatmetric}, \ref{thm-prodtan}, \ref{thm-potential} and \ref{thm-EF} that equipped to
  $\mathcal{M}_{m,n}$ 
there is an infinite-dimensional Frobenius manifold structure.
More precisely, for the Frobenius manifold the flat metric $\langle\,,\,\rangle_\eta$ is
 given in \eqref{metric}, the product is defined by \eqref{prodtan}, the unity vector filed
 $\vec{e}$ is given in \eqref{vece}, the potential function $\mathcal{F}$ is given in \eqref{potential} and the Euler
  vector filed $\vec{E}$ is given in \eqref{vecE}.

\begin{rem}
Recall the definition of $\mathcal{M}_{m,n}$ in Section~\ref{sec-M}. Suppose that the condition $(C3)$ is replace by a general setting:
\begin{itemize}
  \item[(C3)'] The winding number of $\zeta(z)$ around $0$ is an arbitrary positive integer $r$, such that $\zeta(z)^{1/r}$ maps
  the circle $\Gamma$ biholomorphicly to a simple smooth curve $\Sigma$ around $0$.
\end{itemize}
Then we can still carry out the above approach and obtain a Frobenius manifold structure on $\mathcal{M}_{m,n}$ under certain assumptions, with $\zeta(z)$
being replaced by $\zeta(z)^{1/r}$ when introducing the flat coordinates $t_i$.
In particular, if we take
\[
r=2, ~~ \varphi=0, ~~ m=2m', ~~n=2n'
\]
with arbitrary positive integers $m'$ and $n'$,  and assume $(a(z),\hat{a}(z))$ to be a pair of even functions of $z$, then the
Frobenius manifold $\mathcal{M}_{2m',2n'}$ coincides with the one constructed in \cite{WX}  underlying the two-component BKP hierarchy.
\end{rem}

\section{Proof of Main Theorem 2}

In this section, we want to clarify the relationship between the infinte-dimensional Frobenius manifold $\mathcal{M}_{m,n}$ and
the universal Whitham hierarchy.

\subsection{The universal Whitham hierarchy and its bihamiltonian structures} As mentioned before, a general version of  the universal Whitham hierarchy was investigated by Krichever  \cite{Kr88, Kr}, and now we only focus on the case of meromorphic functions with a fixed pole at infinity and a movable pole. More precisely, let us consider two merophorphic functions of $z$ on the Riemann sphere of the form
\begin{align}\label{rela}
&\lambda(z)=z+\sum_{i\ge1}v_{i}(x)(z-\vp(x))^{-i}, \quad
 \hat{\lambda}(z)=\sum_{i\ge-1}\hat{v}_{i}(x)(z-\vp(x))^{i},
\end{align}
with $x\in S^1$ being a loop parameter.
In this paper by the (special) universal Whitham hierarchy we mean the following system of evolutionary equations:
\begin{align}
  &\frac{\p \ld(z)}{\p s_k}=\left[(\lambda(z)^k)_{+},\ld(z)\right],\quad
  \frac{\p\hat{\ld}(z)}{\p s_k}=\left[(\lambda(z)^k)_{+},\hat{\ld}(z)\right],\label{disphir1}\\
  &\frac{\p \ld(z)}{\p
  \hat{s}_k}=\left[-(\hat{\lambda}(z)^k)_{-},\ld(z)\right],\quad
  \frac{\p\hat{\ld}(z)}{\p\hat{s}_k}=\left[-(\hat{\lambda}(z)^k)_{-},\hat{\ld}(z)\right],
  \label{disphir2}
\end{align}
where the Lie bracket $[\,,\,]$ reads
\begin{equation}\label{lieb}
\left[f,g\right]:=\frac{\p f}{\p z} \frac{\p g}{\p x}-\frac{\p
g}{\p z} \frac{\p f}{\p x}
\end{equation}
and  $k=1,2,3,\dots$.
Here it is omitted some
additional flows constructed via logarithm functions in the general version of \cite{Kr}, for they do not concern to what we study below.

Observe that, when $|z|\to\infty$ the function $\ld(z)$ in \eqref{rela} can also be expanded as
$\lambda(z)=z+\sum_{i\ge1}\tilde{v}_i z^{-i}$, thus the flows $\p\ld(z)/\p s_k$ in \eqref{disphir1} indeed compose
the dispersionless KP hierarchy. That is to say, the hierarchy \eqref{disphir1}--\eqref{disphir2} is an extension of the
dispersionless KP hierarchy.
In fact, the universal Whitham hierarchy admits other reductions and the corresponding Frobenius manifolds/WDVV equations have been
constructed in the context of Landau-Ginsburg topological models. For example, when $\hat{\ld}(z)\to0$ and $(\ld(z)^{m})_-=0$
with some positive integer $m$, the  hierarchy \eqref{disphir1}--\eqref{disphir2} is reduced to the dispersionless Gelfand-Dickey
hierarchy, and its bihamiltonian structure was involved by Dubrovin \cite{Du-Whith} in the $A$-series Landau-Ginsburg topological
models. When $\ld(z)^{m}=\hat{\ld}(z)^n$ with positive integers $m$ and $n$, it was achieved by Krichever a residue
formula for the solution of WDVV equations (see Theorem~5.6 in  \cite{Kr} for details, and see also \cite{AK-MPL, AK-CMP}),
which can be regarded as the potential function of a Frobenius manifold. We hope that our construction below for the universal Whitham hierarchy
\eqref{disphir1}--\eqref{disphir2} would help us to understand such works.


Let us recall the bihamiltonian structures derived in \cite{WZ} for the universal Whitham hierarchy \eqref{disphir1}--\eqref{disphir2}.
Given two arbitrary positive integers $m$ and $n$, we consider the loop space, denoted as $\mathcal{L}\mathcal{M}_{m,n}$, of smooth maps
from the unit circle $S^1$ to the manifold $\mathcal{M}_{m,n}$ defined in \eqref{sec-M}.
A point in the loop space $\mathcal{L}\mathcal{M}_{m,n}$ is written as
\[
\vec{a}=(a,\hat{a})=\left(z^{m}+\sum_{i\leq m-2}\mathit{a}_{i}(x)(z-\varphi(x))^{i},\sum_{i\geq -n} \hat{\mathit{a}}_{i}(x)(z-\varphi(x))^{i}\right).
\]
At this point, the tangent space $T_{\vec{a}}\mathcal{L}\mathcal{M}_{m,n}$ and the
cotangent space $T_{\vec{a}}^\ast\mathcal{L}\mathcal{M}_{m,n}$ of the loop space take a similar form as \eqref{TM} and \eqref{TstaM} respectively.

On the loop space $\mathcal{L}\mathcal{M}_{m,n}$, we  introduce a ring of formal differential polynomials as
\[
\mathcal{A}:=C^\infty\left(\mathbf{a}\right)[[\p_x \mathbf{a}, {\p_x}^2\mathbf{a}, \dots ]],
\]
where
\[
\mathbf{a}(x)=(\vp(x), a_{m-2}(x),a_{m-3}(x),\dots,\hat{a}_{-n}(x),\hat{a}_{-n+1}(x),\dots).
\]
Let us consider the quotient space $\mathscr{F}:=\mathcal{A}/\p_x \mathcal{A}$, whose elements are called local functionals written in the form
\[
F=\int
f\left(\mathbf{a}, \p_x \mathbf{a}, {\p_x}^2\mathbf{a}, \dots\right)\,d
x\in\mathscr{F}, \quad f\in\mathcal{A}.
\]
Given a local functional $F\in\mathscr{F}$, its variational gradient at $\vec{a}$ means a
cotangent vector $d F\in T_{\vec{a}}^\ast\mathcal{L}\mathcal{M}_{m,n}$ such that
\[
\dt F=\int \la d F, \dt \vec{a}\ra \,d
x,
\]
with the nondegenerate pairing $\la\,,\,\ra$ given in \eqref{pairing}.

\begin{prop} [\cite{WZ}] \label{thm-eKPham}
For any positive integers $m$ and $n$,
there is a bihamiltonian structure on $\mathcal{L}\mathcal{M}_{m,n}$ given by the follows two compatible Poisson brackets
\begin{equation}\label{pb}
 \big\{F,H\big\}_\nu=\int \la d F,\mathcal{P}_\nu\cdot d H\ra\, d x, \quad
 \nu=1,2.
\end{equation}
Here the Poisson tensors $\mathcal{P}_\nu: T_{\vec{a}}^\ast\mathcal{L}\mathcal{M}_{m,n}\to T_{\vec{a}} \mathcal{L}\mathcal{M}_{m,n}$ read
\begin{align}
\mathcal{P}_{1}\cdot \vec{\omega}=&\big([\omega,a]_{-}+[\hat{\omega},\hat{a}]_{-}-[\omega_{-}+\hat{\omega}_{-},a],
-[\omega,a]_{+}-[\hat{\omega},\hat{a}]_{+}+[\omega_{+}+\hat{\omega}_{+},\hat{a}]\big), \label{P1}
\\
\mathcal{P}_{2}\cdot \vec{\omega}=&\big(([\omega,a]_{-}+[\hat{\omega},\hat{a}]_{-})a-[(\omega a+\hat{\omega}\hat{a})_{-},a]-\sigma a', \nn\\
&-([\omega,a]_{+}+[\hat{\omega},\hat{a}]_{+})\hat{a}+[(\omega a+\hat{\omega}\hat{a})_{+},\hat{a}]-\sigma \hat{a}'\big), \label{P2}
\end{align}
with  the Lie bracket $[\ , \ ]$ given in \eqref{lieb} and
\begin{equation}\label{sg}
\sigma=\frac{1}{m}\res_{z=\varphi}([\omega,a]+[\hat{\omega},\hat{a}])d z.
\end{equation}
Moreover, let
\[
a=\ld(z)^m, \quad \hat{a}=\hat{\ld}(z)^n,
\]
then the universal Whitham hierarchy \eqref{disphir1}--\eqref{disphir2} can be represented in a bihamiltonian recursive form as
    \begin{align}\label{hamrecur}
    &\frac{\p F}{\p s_{k}}=\{F,H_{k+m}\}_{1}=\{F,H_{k}\}_{2},\quad
    \frac{\p F}{\p \hat{s}_{k}}=\{F,\hat{H}_{k+n}\}_{1}=\{F,\hat{H}_{k}\}_{2},
    \end{align}
    where the Hamiltonian functionals are given by
\begin{align}\label{}
H_{k}=-\frac{m}{k}\int \res_{z=\infty}\ld(z)^k\,d x, \quad  \hat{H}_{k}=\frac{n}{k}\int \res_{z=\vp}\hat{\ld}(z)^k\,d x,
\end{align}
with $k=1,2,3,\dots$.
\end{prop}

\begin{rem}
 A dispersionful version of the hierarchy  \eqref{disphir1}--\eqref{disphir2}
was proposed in \cite{SB, WZ} by using scalar pseudo-differential operators, which is an extension of the full KP hierarchy. In analogue of the KP hierarchy, such an extension of it can be represented equivalently as a bilinear equation of Baker-Akhiezer functions \cite{LW}, and
it is equipped with a series of bihamiltonian structures \cite{WZ} whose dispersionless limits are given in the above proposition.
\end{rem}

\subsection{The bihamiltonian structure induced from $\mathcal{M}_{m,n}$}

 In order to stduy the relation between the infinite-dimensional Frobenius manifold  $\mathcal{M}_{m,n}$ and the universal Whitham hierarchy \eqref{disphir1}--\eqref{disphir2}, we
continue to study the intersection form on $\mathcal{M}_{m,n}$.

\begin{lem}
The Euler vector field $\vec{E}$ defined in \eqref{vecE} can be represented in Laurent series as
\begin{equation}\label{vecE2}
\vec{E}=\left(a(z)-\frac{za'(z)}{m},\hat{a}(z)-\frac{z\hat{a}'(z)}{m}\right).
\end{equation}
\end{lem}
\begin{proof} It suffices to show
\begin{equation}\label{}
\vec{E}(\zeta(z))=\zeta(z)-\frac{z\zeta'(z)}{m},\quad \vec{E}(\ell(z))=\ell(z)-\frac{z\ell'(z)}{m}.
\end{equation}
Firstly, substituting $\zeta=\zeta(z)$ into \eqref{zzeta}, one has
\[
z=\sum_{i\in\Z}t_{i}\zeta(z)^{i},
\]
whose derivative with respect to $z$ is
    $$
    1=\sum_{i\in\Z}i t_{i}\zeta(z)^{i-1}\zeta'(z).
    $$
Thus, by using \eqref{zetaellt} one obtains
$$
    \zeta(z)-\frac{z\zeta'(z)}{m}=\sum_{i\in\Z}i t_{i}\zeta(z)^{i}\zeta'(z)- \sum_{i\in\Z}t_{i}\zeta(z)^{i}
     \frac{\zeta'(z)}{m} = \sum_{i\in\mathbb{Z}}\left(i-\frac{1}{m}\right)t_{i}\zeta(z)^{i}\zeta'(z)=\vec{E}(\zeta(z)).
$$
With the same method, by using \eqref{zetaellh} and \eqref{zetaellhh} one gets
    $$
    \begin{aligned}
    \left(\ell(z)-\frac{z\ell'(z)}{m}\right)_{+}=&\sum_{j=11}^{m-1}\frac{i+1}{m}h_{i}(\ell'(z)\chi(z)^{-i})_{+}=\vec{E}(\ell(z)_{+}),\\
    \left(\ell(z)-\frac{z\ell'(z)}{m}\right)_{-}=&-\sum_{i\ge 0}^{n}\left(\frac{1}{m}+\frac{i}{n}\right)\hat{h}_{i}
    (\ell'(z)\hat{\chi}(z)^{-i})_{-} =\vec{E}(\ell(z)_{-}).\\
    \end{aligned}
    $$
Therefore the lemma is proved.
\end{proof}

Recall the generating functions \eqref{rpgener01} and \eqref{rpgener02} of covectors on $T_{\vec a}^\ast\mathcal{M}_{m,n}$. With the help of the Euler vector field, we define the intersection form on the contangent space by
\begin{equation}\label{intform}
\left( d\alpha(p),d\beta(q)\right)^{\ast}:=\mathnormal{i}_{\vec{E}}(d\alpha(p)\star d\beta(q)),\quad \alpha, \beta\in \{a,\hat{a}\}.
\end{equation}

\begin{prop}
The intersection form \eqref{intform} can be represented as
\begin{equation}\label{}
 \left( d\alpha(p),d\beta(q) \right)^{\ast}=\frac{\alpha'(p) \beta(q)}{p-q}+\frac{\beta'(q)\alpha(p)}{q-p}+\frac{\alpha'(p)
 \beta'(q)}{m},\quad \alpha,\beta\in\{a,\hat{a}\}.
\end{equation}
 \end{prop}
\begin{proof}
    For any $\alpha,\beta\in\{a,\hat{a}\}$, by using \eqref{daxi}, \eqref{prodcot} and \eqref{vecE2} one has
    $$
    \begin{aligned}
    \left\langle d\alpha(p)\star d\beta(q),\vec{E}(z)\right\rangle=&\frac{\beta'(q)}{q-p}\left\langle d\alpha(p),
    \vec{E}(z)\right\rangle+\frac{\alpha'(p)}{p-q}\left\langle d\beta(q),\vec{E}(z)\right\rangle\\
    =&\frac{\beta'(q)}{q-p}\left(\alpha(p)-\frac{p\alpha'(p)}{m}\right) +\frac{\alpha'(p)}{p-q}\left(\beta(q)-\frac{q\beta'(q)}{m}\right)\\
    =&\frac{\alpha'(p) \beta(q)}{p-q}+\frac{\beta'(q)\alpha(p)}{q-p}+\frac{\alpha'(p) \beta'(q)}{m}.
    \end{aligned}
    $$
Thus the proposition is proved.
\end{proof}

Similar as before, based on the nondegenerate pairing \eqref{pairing}, there is a linear map
\begin{equation}\label{gmap}
g: \, T^{\ast}_{\vec{a}}\mathcal{M}_{m,n}\to T_{\vec{a}}\mathcal{M}_{m,n}
\end{equation}
defined by
\begin{equation}\label{}
    \langle\vec{\omega}_{1},g\cdot\vec{\omega}_{2}\rangle=\left( \vec{\omega}_{1},\vec{\omega}_{2}\right)^{\ast},
    \quad\vec{\omega}_{1},\vec{\omega}_{2}\in T^{\ast}_{\vec{a}}\mathcal{M}_{m,n}.
\end{equation}
\begin{lem}\label{thm-gom}
 The map $g$ defined above is a bijection. More precisely, for any $(\omega(z),\hat{\omega}(z))\in T^{\ast}_{\vec{a}}\mathcal{M}_{m,n}$ it holds that  $g\cdot(\omega(z),\hat{\omega}(z))=(\xi(z), \hat{\xi}(z))$ with
\begin{align}\label{gom}
   \xi(z)=& a'(z)(a(z)\omega(z)+\hat{a}(z)\hat{\omega}(z))_{-} -a(z)(a'(z)\omega(z)+\hat{a}'(z)\hat{\omega}(z))_{-}+\rho a'(z), \\
    \hat{\xi}(z)=&-\hat{a}'(z)(a(z)\omega(z)+\hat{a}(z)\hat{\omega}(z))_{+} +\hat{a}(z)(a'(z)\omega(z)
    +\hat{a}'(z)\hat{\omega}(z))_{+}+\rho\hat{a}'(z), \label{gomh}
\end{align}
where
\[
\rho=\frac{1}{m}\res_{z=\varphi}\left( a'(z)\omega(z)+\hat{a}'(z)\hat{\omega}(z) \right)d z.
\]
\end{lem}
\begin{proof}
Firstly, one can verify \eqref{gom} and \eqref{gomh} in the same way as in Lemma~\ref{thm-etaom}, with $\langle\,,\,\rangle^\ast$ replaced by  $(\,,\,)^\ast$.
Now let us show that the map $g$ is bijective. In fact, according to \eqref{gom} and \eqref{gomh} one has
	\begin{align*}
		\hat{a}'\xi-a'\hat{\xi}=&a'\hat{a}'(a \om+\hat{a}\hat{\om})- \hat{a}'a(a' \om+\hat{a}'\hat{\om})_- - a'\hat{a} (a' \om+\hat{a}'\hat{\om})_+ \\ =&(\hat{a}'a-a'\hat{a})\left( (a'\omega)_{+}-(\hat{a}'\hat{\omega})_{-} \right),
	\end{align*}
which leads to (recall the third inequality in \eqref{M2})
	\begin{align}
		(a'\omega)_{+}=\left( \frac{\hat{a}'\xi-a'\hat{\xi}}{\hat{a}'a-a'\hat{a}}\right)_{+},
\quad (\hat{a}'\hat{\omega})_{-}=-\left( \frac{\hat{a}'\xi-a'\hat{\xi}}{\hat{a}'a-a'\hat{a}}\right)_{-}. \label{a'omega}
	\end{align}
To simplify notations,  let us denote
\[
K(z)=\frac{\hat{a}'(z)\xi(z)-a'(z)\hat{\xi}(z)}{\hat{a}'(z)a(z)-a'(z)\hat{a}(z)}.
\]
By using $\frac{1}{a'}\in (z-\vp)^{-m+1}\mathcal{H}_\vp^-$  and $\frac{1}{\hat{a}'}\in (z-\vp)^{n+1}\mathcal{H}_\vp^+$, one obtains
	\begin{align}
		\omega=&(\omega)_{\ge -m+1}=\left(\frac{1}{a'}(a'\omega)_{+}\right)_{\ge -m+1}
=\left(\frac{1}{a'}K_{+}\right)_{\ge -m+1}=\left(\frac{K}{a'} \right)_{\ge -m+1},\label{om}\\
			\hat{\omega}=
&(\hat{\omega})_{\le n} =\left(\frac{1}{\hat{a}'}(\hat{a}'\hat{\omega})_{-}\right)_{\le n}
=-\left(\frac{1}{\hat{a}'}K_{-}\right)_{\le n}=-\left(\frac{K}{\hat{a}'}\right)_{\le n}.\label{omh}
	\end{align}

Conversely, for any $(\xi, \hat{\xi})\in T_{\vec{a}}\mathcal{M}_{m,n}$, one has $(\om, \hat{\om})$
by the formulae \eqref{om} and \eqref{omh}. Now let us compute $(\psi, \hat{\psi}):=g\cdot(\om, \hat{\om})$. To this end, one has
\begin{align*}
\rho=&\frac{1}{m}\res_{z=\varphi}(a'\omega+\hat{a}'\hat{\omega})dz \\
=&\frac{1}{m}\res_{z=\varphi}\left( a' \left(\frac{K}{a'} \right)_{\ge -m+1} -\hat{a}' \left(\frac{K}{\hat{a}'}\right)_{\le n} \right)d z \\
=&\frac{1}{m}\res_{z=\varphi}\left( a'  \frac{K}{a'}  -\hat{a}' \frac{K}{\hat{a}'}  \right)d z - \res_{z=\vp} \frac{ K}{a'} (z-\vp)^{m-1}d z \\
=& - \res_{z=\vp} \frac{ K}{a'} (z-\vp)^{m-1}d z,
\\
(a\omega)_{+}=&\left(a \left(\frac{K}{a'}\right)_{\ge -m+1} \right)_+  =\left(a \frac{ K}{a'}  \right)_+ - \res_{z=\vp} \frac{ K}{a'} (z-\vp)^{m-1}d z
\\
=&\left(\frac{\xi}{a'}+\frac{\hat{a}\xi-a\hat{\xi}}{\hat{a}'a-a'\hat{a}}\right)_{+} +\rho
=\left(\frac{\hat{a}\xi-a\hat{\xi}}{\hat{a}'a-a'\hat{a}}\right)_{+} +\rho,
\\
(\hat{a}\hat{\omega})_{-}=&-\left(\hat{a}\left(\frac{K}{\hat{a}'}\right)_{\le n} \right)_- =-\left(\hat{a} \frac{K}{\hat{a}'}  \right)_- \\
=&-\left(\frac{\hat{a}\xi-a\hat{\xi}}{\hat{a}'a-a'\hat{a}}+\frac{\hat\xi}{\hat{a}'}\right)_{+}
=-\left(\frac{\hat{a}\xi-a\hat{\xi}}{\hat{a}'a-a'\hat{a}}\right)_{-}.
	\end{align*}
Hence, by using \eqref{gom}, \eqref{gomh} and \eqref{a'omega}, we obtain
	\begin{align*}
\psi=&a\left( (a'\omega)_{+}-(\hat{a}'\hat{\omega})_{-}\right) -a'\left( (a\omega)_{+}-(\hat{a}\hat{\omega})_{-}-\rho\right)
=a \frac{\hat{a}'\xi-a'\hat{\xi}}{\hat{a}'a-a'\hat{a}}- a' \frac{\hat{a}\xi-a\hat{\xi}}{\hat{a}'a-a'\hat{a}} =\xi,
\\
\hat{\psi}=& \hat{a}\left((a'\omega)_{+}- (\hat{a}'\hat{\omega})_{-}\right) -\hat{a}'\left((a\omega)_{+}
-(\hat{a}\hat{\omega})_{-}-\rho\right) = \hat{a} \frac{\hat{a}'\xi-a'\hat{\xi}}{\hat{a}'a-a'\hat{a}}- \hat{a}'
\frac{\hat{a}\xi-a\hat{\xi}}{\hat{a}'a-a'\hat{a}} =\hat{\xi}.
	\end{align*}
Thus the lemma is proved.
\end{proof}

With the help of the bijection $g$, let us introduce another metric on $\mathcal{M}_{m,n}$ as
\begin{equation}\label{metric2}
\langle \p_1,\p_2\rangle_{g}:=\la
  g^{-1}(\p_1),\p_2\ra=\left( g^{-1}(\p_1),g^{-1}(\p_2)\right)^\ast, \quad \p_1, \p_2\in T_{\vec{a}}\mathcal{M}_{m,n}.
\end{equation}

\begin{thm} \label{thm-PBmetric}
The compatible Poisson brackets \eqref{pb} are the ones induced by the metrics $\la\,,\,\ra_\eta$ and  $\la\,,\,\ra_g$ given in \eqref{metric} and in \eqref{metric2} respectively.
\end{thm}

\begin{proof}
Denote $\vec{\omega}_x=(\omega_{x}, \hat{\omega}_{x})=\left(\frac{\p \omega}{\p x}, \frac{\p \hat{\omega}}{\p x}\right)$ for $\vec{\omega}\in T_{\vec{a}}^\ast\mathcal{L}\mathcal{M}_{m,n}$. Observe that the formulae \eqref{P1} and \eqref{P2} can be represented as
\[
\mathcal{P}_\nu\cdot \vec{\omega}=P_\nu(\vec{\omega}_x)+Q_\nu(\vec{\omega}), \quad \nu=1, 2,
\]
where $P_\nu$ and $Q_\nu$ are functions depending linearly on $\vec{\om}_x$ and on $\vec{\om}$ respectively.
According to the theory of hamiltonian structures of hydradynamic type \cite{DN}, we only need to compare the functions
$P_{\nu}\, (\nu=1,2)$ with $\eta\cdot\vec{\omega}_x$ and $g\cdot\vec{\omega}_x$ (recall \eqref{eta} and \eqref{gmap}).
More explicitly, by using Lemmas~\ref{thm-etaom} and \ref{thm-gom} it is straightforward to check
\begin{align*}
P_{1}(\vec{\omega}_x)=&\big(-(a'\omega_{x})_{-}-(\hat{a}'\hat{\omega}_{x})_{-} +a'(\omega_{x})_- +a'(\hat{\omega}_{x})_-, \\
&\qquad
(a'\omega_{x})_{+}+(\hat{a}'\hat{\omega}_{x})_{+}-\hat{a}(\omega_{x})_+ -\hat{a}(\hat{\omega}_{x})_+\big) \\
=&\eta\cdot\vec{\omega}_{x},
\\
P_{2}(\vec{\omega}_x)=&\bigg( -a(a'\omega_{x}+\hat{a}'\hat{\omega}_{x})_{-}+ a'(a\omega_{x}+\hat{a}\hat{\omega}_{x})_{-} -\frac{a'}{m}\res_{z=\varphi}(-a'\omega_{x}- \hat{a}'\hat{\omega}_{x})dz,\\
& \qquad \hat{a}(a'\omega_{x}+\hat{a}'\hat{\omega}_{x})_{+}-\hat{a}'(a\omega_{x}+\hat{a}\hat{\omega}_{x})_{+} -\frac{\hat{a}'}{m}\res_{z=\varphi}(-a'\omega_{x}- \hat{a}'\hat{\omega}_{x})dz\bigg)\\
=&g\cdot\vec{\omega}_{x}.
\end{align*}
Therefore the theorem is proved.
\end{proof}

So we complete the proof of Main Theorem 2.


\section{Concluding remarks}

In this paper we have constructed a class of infinite-dimensional Frobenius manifolds underlying the universal Whitham hierarchy \eqref{disphir1}--\eqref{disphir2},
whose dispersionful analogue is an extension of the KP hierarchy. More exactly, on such a Frobenius manifold $\mathcal{M}_{m,n}$
with positive integers $m$ and $n$, its flat coordinates for the metric, the product with unity vector field, the potential function $\mathcal{F}$,
the Euler vector field and the intersection form are obtained. Moreover, the flat pencil of metrics on the Frobenius manifold is
 proved to induce the bihamiltonian structure \eqref{pb} for the universal Whitham hierarchy \eqref{disphir1}--\eqref{disphir2}. In other words, the method
 initiated in \cite{CDM} and developed in \cite{WX, WZuo} is applied successfully to a wider range of concrete examples.

Similar to the finite-dimensional case, if we take the following Hamiltonian functionals given by the derivatives of the potential function as
\[
\mathcal{H}_{u,0}=\int \left.\frac{\p \mathcal{F}}{\p u}\right|_{v\mapsto v(x)} d x,\quad u, v\in\mathbf{t}\cup\mathbf{h}\cup\hat{\mathbf{h}},
\]
then the Hamiltonian equations given by the first Poisson bracket in \eqref{pb} are up to constant factors with those flows
of the nondecendant level in the principal hierarchy associated to $\mathcal{M}_{m,n}$.
We expect that the bihamiltonian recursion relation in \eqref{hamrecur} might help us to obtain the complete principal hierarchy.
Furthermore, a tau function of the principal hierarchy is expected to be introduced, which may help us to understand the
tau function with that solving the string equation for the Whitham hierarchy  studied in \cite{AK-CMP, Kr}. We will consider it in other occasions.

\medskip
\noindent{\bf Acknowledgements}.
The second and the third authors thank Professors Boris Dubrovin and  Youjin Zhang for their advising, and
they also thank Professor Ian Strachan for helpful discussions on this topic.
The author C.-Z.Wu is partially supported by NSFC (No.11771461, No.11831017), and the author  D.Zuo is
 partially supported by NSFC (No.11671371, No.11871446) and Wu Wen-Tsun Key Laboratory of
 Mathematics, USTC, CAS.


\end{document}